\documentclass[12pt]{article}
\usepackage[english]{babel}
\usepackage{amsmath}
\usepackage{amsthm}
\usepackage{amssymb}
\usepackage{authblk}
\usepackage{fullpage}
\usepackage{cite}
\usepackage[colorlinks=true]{hyperref}
\usepackage{tikz}

\newtheorem{thm}{Theorem}[section]

\newtheorem*{thm*}{Theorem}
\newtheorem{cor}{Corollary}[section]

\newtheorem{lem}{Lemma}[section]

\newtheorem{prop}{Proposition}[section]

\newtheorem*{prop*}{Proposition}
\theoremstyle{definition}
\newtheorem{defn}{Definition}[section]

\theoremstyle{remark}
\newtheorem{rem}{Remark}[section]

\numberwithin{equation}{section}

\begin{document}

\title{The Wasserstein distance of order $1$ for quantum spin systems on infinite lattices}

\author[1]{Giacomo De Palma}
\author[2]{Dario Trevisan}
\affil[1]{Department of Mathematics, University of Bologna, Bologna, Italy}
\affil[2]{Department of Mathematics, University of Pisa, Pisa, Italy}

\maketitle

\begin{abstract}
We propose a generalization of the Wasserstein distance of order $1$ to quantum spin systems on the lattice $\mathbb{Z}^d$, which we call specific quantum $W_1$ distance.
The proposal is based on the $W_1$ distance for qudits of [De Palma \emph{et al.}, IEEE Trans. Inf. Theory 67, 6627 (2021)] and recovers Ornstein's $\bar{d}$-distance for the quantum states whose marginal states on any finite number of spins are diagonal in the canonical basis.
We also propose a generalization of the Lipschitz constant to quantum interactions on $\mathbb{Z}^d$ and prove that such quantum Lipschitz constant and the specific quantum $W_1$ distance are mutually dual.
We prove a new continuity bound for the von Neumann entropy for a finite set of quantum spins in terms of the quantum $W_1$ distance, and we apply it to prove a continuity bound for the specific von Neumann entropy in terms of the specific quantum $W_1$ distance for quantum spin systems on $\mathbb{Z}^d$.
Finally, we prove that local quantum commuting interactions above a critical temperature satisfy a transportation-cost inequality, which implies the uniqueness of their Gibbs states.
\end{abstract}

\section{Introduction}

Let $\mathcal{X}$ be a finite set endowed with the distance $D$ and let $\mu$ and $\nu$ be probability distributions on $\mathcal{X}$.
A \emph{coupling} between $\mu$ and $\nu$ is a probability distribution on two copies of $\mathcal{X}$ with marginal distributions equal to $\mu$ and $\nu$, respectively.
The theory of optimal mass transport considers $\mu$ and $\nu$ as distributions of a unit amount of mass, and any coupling $\pi$ prescribes a plan to transform the distribution $\mu$ into the distribution $\nu$, \emph{i.e.}, $\pi(x,y)$
is the amount of mass that is moved from $x$ to $y$ for any $x,\,y\in\mathcal{X}$.
Assuming that the cost of moving a unit of mass from $x$ to $y$ is equal to $D(x,y)$, the cost of the coupling $\pi$ is equal to $\mathbb{E}_{(X,Y)\sim\pi}D(X,Y)$.
The Monge--Kantorovich distance between $\mu$ and $\nu$ is given by the minimum cost among all the couplings between $\mu$ and $\nu$ \cite{monge1781memoire, kantorovich1942translocation, vershik2013long}.
Such distance is also called earth mover's distance or Wasserstein distance of order $1$, often shortened to $W_1$ distance.
The exploration of the theory of optimal mass transport has led to the creation of an extremely fruitful field in mathematical analysis, with applications ranging from differential geometry and partial differential equations to machine learning \cite{villani2008optimal, ambrosio2008gradient, peyre2019computational,vershik2013long}.

The Hamming distance constitutes a natural choice for the distance $D$ when $\mathcal{X}$ is a set of finite strings over an alphabet.
The $W_1$ distance with respect to the Hamming distance is called Ornstein's $\bar{d}$-distance and was first considered in \cite{ornstein1973application}, together with its extension to stationary stochastic processes. Originally introduced as a tool for the classification of stationary processes in ergodic theory, since then it has found further applications in probability theory, such as the statistical estimation of processes \cite{ornstein1990sampling, ornstein1994d, fernandez2002markov, csiszar2010rate, o2021estimation},  information theory, such as coding theorems for a large class of discrete noisy channels with memory and rate distortion theory \cite{gray1975generalization, gray2011entropy} and recently also machine learning, as a peculiar case of Wasserstein auto-encoders \cite{choi2019ornstein, choi2021learning}.

Ref. \cite{de2021quantum} proposed a generalization of the $W_1$ distance to the space of the quantum states of a finite set of qudits or spins, called quantum $W_1$ distance.
The generalization is based on the notion of neighboring quantum states.
Two quantum states of a finite set of qudits are neighboring if they coincide after discarding one qudit.
The quantum $W_1$ distance proposed in Ref. \cite{de2021quantum} is the distance induced by the maximum norm that assigns distance at most $1$ to any couple of neighboring states.
Such quantum $W_1$ distance recovers Ornstein's $\bar{d}$-distance in the case of quantum states diagonal in the canonical basis and inherits most of its properties.
The quantum $W_1$ distance has found several applications in quantum information theory.
In the context of statistical mechanics of quantum spin systems, a connection with quantum speed limits \cite{hamazaki2022speed} has been found.
Furthermore, transportation-cost inequalities have been proved, which upper bound the square of the quantum $W_1$ distance between a generic quantum state and the Gibbs state of a local quantum commuting Hamiltonian with the relative entropy between the same states \cite{de2022quantum}.
Such inequalities have been applied to prove the equivalence between the microcanonical and the canonical ensembles of quantum statistical mechanics \cite{de2022quantum} and to prove limitations of variational quantum algorithms \cite{de2022limitations,chou2022limitations}.
Moreover, the quantum $W_1$ distance has been applied to quantify the complexity of quantum circuits \cite{li2022wasserstein}.
In the context of quantum state tomography, the quantum $W_1$ distance has been employed as quantifier of the quality of the learned quantum state, and the transportation-cost inequalities have led to an efficient algorithm to learn Gibbs states of local quantum commuting Hamiltonians \cite{rouze2021learning,maciejewski2021exploring}.
In the context of quantum machine learning, the quantum $W_1$ distance has been employed as cost function of the quantum version of generative adversarial networks \cite{kiani2022learning,herr2021anomaly,anschuetz2022beyond,coyle2022machine}.
Furthermore, the quantum $W_1$ distance has been applied in the context of differential privacy of a quantum computation \cite{hirche2022quantum,angrisani2022differential}.
Finally, the quantum $W_1$ distance has been extended to general composite systems \cite{duvenhage2022quantum}, which include the case of a finite tensor product of $C^*$ algebras, but also provides a way to define a quantum $W_1$ distance between quantum channels.

\subsection{Our contribution}
In this paper we propose a generalization of the $W_1$ distance to quantum spin systems on the lattice $\mathbb{Z}^d$ \cite{bratteli2013operatorI,bratteli2013operatorII,naaijkens2017quantum,alicki2001quantum} based on the quantum $W_1$ distance of Ref. \cite{de2021quantum}.
Quantum spin systems on infinite lattices play a key role in quantum statistical mechanics since they provide a model to study the thermodynamic limit of infinite size of the system.
Such limit is necessary to define phase transitions and to identify the properties of the system that are independent on boundary effects and boundary conditions, and more generally to make a clear distinction between the local and the global properties of the system.

We define the specific quantum $W_1$ distance between two translation-invariant states as the limit of the distance between their marginal states on an hypercube divided by the volume of the hypercube for the volume of the hypercube tending to infinity (\autoref{def:w1}).
Contrarily to the trace distance, the specific quantum $W_1$ distance has an intensive nature that make it suitable to capture the closeness of states that are locally similar but become perfectly distinguishable globally, such as Gibbs states at close but different temperatures.
We provide in \autoref{def:alternative-wass} an equivalent definition of the specific quantum $W_1$ distance that does not require the limit.
We propose a generalization of the Lipschitz constant to quantum interactions on $\mathbb{Z}^d$ (\autoref{def:L}), and we prove in \autoref{thm:duality} that the specific quantum $W_1$ distance and the Lipschitz constant are mutually dual.

We prove in \autoref{prop:recovery} that the specific quantum $W_1$ distance recovers Ornstein's $\bar{d}$-distance in the case of quantum states whose marginal states on a finite number of spins are all diagonal in the canonical basis.
We prove in \autoref{prop:PoincI} a Poincar\'e inequality stating that for any product state, the variance of the local Hamiltonians associated with an interaction grows linearly with the volume.
In \autoref{thm:Gauss}, we prove a Gaussian concentration inequality for the maximally mixed state of a finite set of spins, and we apply it in \autoref{cor:PL} to prove an upper bound to the pressure of a quantum interaction on $\mathbb{Z}^d$ in terms of its Lipschitz constant.

In \autoref{thm:main}, we prove a continuity bound for the von Neumann entropy in terms of the $W_1$ distance.
The bound applies to quantum systems made by a finite number of spins or qudits and states that the difference between the von Neumann entropy of any two quantum states divided by the number of spins is upper bounded by a universal function of the ratio between the $W_1$ distance and the number of spins.
The bound of \autoref{thm:main} contains only intensive quantities, and thanks to this property we apply it to prove a continuity bound for the specific von Neumann entropy in terms of the specific quantum $W_1$ distance (\autoref{thm:mainI}).
\autoref{thm:main} improves \cite[Theorem 1]{de2021quantum}, which is a weaker continuity bound for the von Neumann entropy in terms of the $W_1$ distance.
Contrarily to the bound of \autoref{thm:main}, the bound of \cite[Theorem 1]{de2021quantum} cannot be expressed in terms of only intensive quantities, and therefore such bound would not be sufficient to prove a continuity bound for the specific von Neumann entropy.
Besides the applications to quantum spin systems, \autoref{thm:main} can be useful in quantum Shannon theory in the context of rate-distortion theory, which addresses the problem of determining the maximum compression rate of a quantum state if a certain level of distortion in the recovered state is allowed \cite{barnum2000quantum,devetak2001quantum,devetak2002quantum,chen2008entanglement, datta2012quantum,datta2013quantum,wilde2013quantum,salek2018quantum}.

In the remainder of the paper we apply our $W_1$ distance to study the statistical mechanics of quantum spin systems on infinite lattices.
We propose a definition of $w_1$-Gibbs state as a translation-invariant state such that the $W_1$ distance between its marginal state on a hypercube and the Gibbs state of the local Hamiltonian on the same hypercube scales sublinearly with the volume of the hypercube (\autoref{defn:w1G}).
If an interaction admits a $w_1$-Gibbs state, then such state is unique (\autoref{prop:w1G}) and is an equilibrium state of the interaction (\autoref{prop:w1eq}) in the sense of Kubo--Martin--Schwinger  \cite{bratteli2013operatorII}.
In \autoref{sec:TCI} we consider transportation-cost inequalities for interactions on the quantum spin lattice $\mathbb{Z}^d$.
Such inequalities imply the uniqueness of the Gibbs state of the interaction (\autoref{thm:uniqueness}) and a continuity bound for the specific entropy in terms of the specific relative entropy with respect to the Gibbs state (\autoref{prop:ss}).
Finally, we prove that transportation-cost inequalities are satisfied by interactions that contain only terms acting on single spins (\autoref{cor:prodTCI}) and geometrically local commuting interactions above a critical temperature (\autoref{thm:TCI} and \autoref{thm:TCI2}).

The paper is structured as follows.
In \autoref{sec:defs} we introduce quantum spin systems on the lattice $\mathbb{Z}^d$ and in \autoref{sec:W1} we present the quantum $W_1$ distance and the quantum Lipschitz constant of Ref. \cite{de2021quantum}.
In \autoref{sec:w1} and \autoref{sec:L} we generalize the quantum $W_1$ distance and the quantum Lipschitz constant, respectively, to quantum spin systems on the lattice $\mathbb{Z}^d$.
In \autoref{sec:duality} we prove the duality between the specific quantum $W_1$ distance and the Lipschitz constant and in \autoref{sec:recovery} we prove that the specific quantum $W_1$ distance that we propose recovers Ornstein's $\bar{d}$-distance.
In \autoref{sec:conc} we prove the quantum Poincar\'e and Gaussian concentration inequalities for product states.
In \autoref{sec:continuity} we prove the continuity bound for the von Neumann entropy in terms of the $W_1$ distance, and in \autoref{sec:w1cont} we prove the continuity bound for the specific entropy in terms of the specific quantum $W_1$ distance.
In \autoref{sec:Gibbs} we introduce the notion of $w_1$-Gibbs state.
In \autoref{sec:TCI} we present and prove the transportation-cost inequalities for Gibbs states.
We conclude in \autoref{sec:persp} presenting some perspective applications of this work.
\autoref{app:W1} recalls some relevant properties of the quantum $W_1$ distance.
\autoref{app:auxproofs} contains some auxiliary proofs, and \autoref{app:aux} contains the proof of the auxiliary lemmas.

\subsection{Related approaches}

Several quantum generalizations of optimal transport distances have been proposed.
One line of research by Carlen, Maas, Datta and Rouz\'e \cite{carlen2014analog,carlen2017gradient,carlen2020non,rouze2019concentration,datta2020relating,van2020geometrical,wirth2022dual} defines a quantum Wasserstein distance of order $2$ from a Riemannian metric on the space of quantum states based on a quantum analog of a differential structure.
Exploiting their quantum differential structure, Refs. \cite{rouze2019concentration,carlen2020non,gao2020fisher} also define a quantum generalization of the Lipschitz constant and of the Wasserstein distance of order $1$.
Alternative definitions of quantum Wasserstein distances of order $1$ based on a quantum differential structure are proposed in Refs. \cite{chen2017matricial,ryu2018vector,chen2018matrix,chen2018wasserstein}.
Refs. \cite{agredo2013wasserstein,agredo2016exponential,ikeda2020foundation} propose quantum Wasserstein distances of order $1$ based on a distance between the vectors of the canonical basis.

Another line of research by Golse, Mouhot, Paul and Caglioti \cite{golse2016mean,caglioti2021towards,golse2018quantum,golse2017schrodinger,golse2018wave, caglioti2019quantum,friedland2021quantum, cole2021quantum, duvenhage2021optimal,bistron2022monotonicity,van2022thermodynamic} arose in the context of the study of the semiclassical limit of quantum mechanics and defines a family of quantum Wasserstein distances of order $2$ built on a quantum generalization of couplings.
Such distances have been generalized to von Neumann algebras \cite{duvenhage2020quadratic,duvenhage2021wasserstein,duvenhage2022extending}.

Ref. \cite{de2021quantumAHP} proposes another quantum Wasserstein distance of order $2$ based on couplings, with the property that each quantum coupling is associated to a quantum channel.
The relation between quantum couplings and quantum channels in the framework of von Neumann algebras has been explored in \cite{duvenhage2018balance}.
The problem of defining a quantum Wasserstein distance of order $1$ through quantum couplings has been explored in Ref. \cite{agredo2017quantum}.

The quantum Wasserstein distance between two quantum states can be defined as the classical Wasserstein distance between the probability distributions of the outcomes of an informationally complete measurement performed on the states, which is a measurement whose probability distribution completely determines the state.
This definition has been explored for Gaussian quantum systems with the heterodyne measurement in Refs. \cite{zyczkowski1998monge,zyczkowski2001monge,bengtsson2017geometry}.

\section{Quantum spin systems on infinite lattices}\label{sec:defs}

In this section we introduce the setting of quantum spin systems on infinite lattices and fix the notation for the paper.
For more details, the reader is encouraged to consult the books \cite{bratteli2013operatorI,bratteli2013operatorII,alicki2001quantum,naaijkens2017quantum}.

\subsection{Algebra and states}
We associate to each $x\in\mathbb{Z}^d$ the single-spin Hilbert space $\mathcal{H}_x = \mathbb{C}^q$.
Let $\mathcal{F}_{\mathbb{Z}^d}$ be the collection of all the finite subsets of $\mathbb{Z}^d$.
We associate to each $\Lambda\in\mathcal{F}_{\mathbb{Z}^d}$ the Hilbert space
\begin{equation}
\mathcal{H}_\Lambda = \bigotimes_{x\in \Lambda}\mathcal{H}_x\,.
\end{equation}
For each $\Lambda\in\mathcal{F}_{\mathbb{Z}^d}$, we denote with $\mathfrak{U}_\Lambda$ the algebra of the linear operators acting on $\mathcal{H}_\Lambda$ equipped with the operator norm, which we denote with $\|\cdot\|_\infty$.
For any $\Lambda'\subseteq \Lambda$, $\mathfrak{U}_{\Lambda'}$ can be canonically identified with a subalgebra of $\mathfrak{U}_\Lambda$.
This identification will always be implicit.

We denote with $\mathcal{O}_\Lambda\subset\mathfrak{U}_\Lambda$ the set of the self-adjoint linear operators acting on $\mathcal{H}_\Lambda$, and with $\mathcal{O}^T_\Lambda\subset\mathcal{O}_\Lambda$ the set of the traceless self-adjoint linear operators acting on $\mathcal{H}_\Lambda$.
We denote with $\mathcal{S}_\Lambda\subset\mathcal{O}_\Lambda$ the set of the quantum states acting on $\mathcal{H}_\Lambda$, \emph{i.e.}, the positive semidefinite linear operators with unit trace, and with $\mathrm{Tr}_\Lambda$ the partial trace over $\mathcal{H}_\Lambda$.
We say that $\rho\in\mathcal{S}_\Lambda$ is a product state if there exists a collection of states $\left\{\rho_x\in\mathcal{S}_x\right\}_{x\in\Lambda}$ such that
\begin{equation}
\rho = \bigotimes_{x\in\Lambda}\rho_x\,.
\end{equation}
Some results of this paper do not require the lattice structure of $\mathbb{Z}^d$ and apply to generic finite spin systems.
If $\Lambda$ is a generic finite set, we still define $\mathcal{H}_\Lambda$, $\mathrm{Tr}_\Lambda$, $\mathfrak{U}_\Lambda$, $\mathcal{O}_\Lambda$, $\mathcal{O}_\Lambda^T$ and $\mathcal{S}_\Lambda$ as above.

The strictly local algebra of the spin lattice $\mathbb{Z}^d$ is
\begin{equation}
\mathfrak{U}_{\mathbb{Z}^d}^{loc} = \bigcup_{\Lambda\in\mathcal{F}_{\mathbb{Z}^d}}\mathfrak{U}_\Lambda\,,
\end{equation}
and is equipped with the norm inherited from the operator norm of each $\mathfrak{U}_\Lambda$.
The quasi-local algebra $\mathfrak{U}_{\mathbb{Z}^d}$ is the completion of $\mathfrak{U}_{\mathbb{Z}^d}^{loc}$ with respect to such norm, which we still denote with $\|\cdot\|_\infty$.
For any (not necessarily finite) $\Gamma\subseteq\mathbb{Z}^d$, we define
\begin{equation}\label{eq:ULambda}
\mathfrak{U}_\Gamma = \overline{\bigcup_{X\in\mathcal{F}_{\mathbb{Z}^d},\,X\subseteq\Gamma}\mathfrak{U}_X}\subseteq\mathfrak{U}_{\mathbb{Z}^d}\,,
\end{equation}
where the closure is with respect to the $\|\cdot\|_\infty$ norm in $\mathfrak{U}_{\mathbb{Z}^d}$.
When $\Gamma$ is finite or $\Gamma=\mathbb{Z}^d$, \eqref{eq:ULambda} is consistent with the previous definitions.
We denote with $\mathcal{O}_\Gamma$ the set of the self-adjoint elements of $\mathfrak{U}_\Gamma$.

A quantum state $\rho$ of the spin lattice $\mathbb{Z}^d$ is a positive linear functional on $\mathfrak{U}_{\mathbb{Z}^d}$ with $\rho(\mathbb{I})=1$.
We denote the set of the quantum states of $\mathbb{Z}^d$ with $\mathcal{S}_{\mathbb{Z}^d}$.
Analogously, for any (not necessarily finite) $\Gamma\subseteq\mathbb{Z}^d$, a quantum state $\rho$ of $\Gamma$ is a positive linear functional on $\mathfrak{U}_\Gamma$ with $\rho(\mathbb{I})=1$.
We denote with $\mathcal{S}_\Gamma$ the set of the quantum states of $\Gamma$.
If $\Gamma$ is finite, this definition is consistent with the previous one by setting for any $A\in\mathfrak{U}_\Gamma$
\begin{equation}
\rho(A) = \mathrm{Tr}_\Gamma\left[\rho\,A\right]\,.
\end{equation}
Let $\Gamma'\subseteq\Gamma\subseteq\mathbb{Z}^d$ and let $\rho\in\mathcal{S}_\Gamma$.
The marginal state $\rho_{\Gamma'}$ of $\rho$ on $\Gamma'$ is the restriction of $\rho$ on $\mathfrak{U}_{\Gamma'}$.
If $\Gamma$ is finite we have
\begin{equation}\label{eq:marginal-marginal}
\rho_{\Gamma'} = \mathrm{Tr}_{\Gamma\setminus\Gamma'}\rho\,.
\end{equation}
Since $\mathfrak{U}_{\mathbb{Z}^d}^{loc}$ is dense in $\mathfrak{U}_{\mathbb{Z}^d}$ by construction, any $\rho\in\mathcal{S}_{\mathbb{Z}^d}$ is completely determined by the collection of its marginal states $(\rho_\Lambda)_{\Lambda\in\mathcal{F}_{\mathbb{Z}^d}}$.
We say that $\rho\in\mathcal{S}_{\mathbb{Z}^d}$ is a product state if for any $\Lambda\in\mathcal{F}_{\mathbb{Z}^d}$ the marginal state $\rho_\Lambda$ is a product state.

We associate to each $x\in\mathbb{Z}^d$ the translation operator $\tau_x$, which is the automorphism of $\mathfrak{U}_{\mathbb{Z}^d}$ that sends $\mathfrak{U}_\Gamma$ to $\mathfrak{U}_{\Gamma+x}$ for any $\Gamma\subseteq\mathbb{Z}^d$.
With some abuse of notation, we denote with $\tau_x$ also the automorphism of $\mathcal{S}_{\mathbb{Z}^d}$ such that for any $\rho\in\mathcal{S}_{\mathbb{Z}^d}$ and any $A\in\mathfrak{U}_{\mathbb{Z}^d}$ we have
\begin{equation}\label{eq:deftau}
(\tau_x\rho)(\tau_x A) = \rho(A)\,.
\end{equation}
With some further abuse of notation, for any $\Gamma\subseteq\mathbb{Z}^d$ we denote with $\tau_x$ also the isomorphism between $\mathcal{S}_\Gamma$ and $\mathcal{S}_{\Gamma+x}$ such that \eqref{eq:deftau} holds for any $\rho\in\mathcal{S}_\Gamma$ and any $A\in\mathfrak{U}_\Gamma$.
We say that $\rho\in\mathcal{S}_{\mathbb{Z}^d}$ is translation invariant if $\tau_x\rho = \rho$ for any $x\in\mathbb{Z}^d$.
We denote with $\mathcal{S}_{\mathbb{Z}^d}^I\subset\mathcal{S}_{\mathbb{Z}^d}$ the set of the translation-invariant quantum states of $\mathbb{Z}^d$.

For each $a\in\mathbb{N}_+^d$, we denote with $\Lambda_a$ the box
\begin{equation}
\Lambda_a = \left\{x\in\mathbb{Z}^d: -a \le x < a\right\}\,,\qquad\left|\Lambda_a\right| = \prod_{i=1}^d 2a_i\,,
\end{equation}
where inequalities between vectors hold for each component.
Given a sequence $\left(a^{(n)}\right)_{n\in\mathbb{N}}\subset\mathbb{N}_+^d$, we say that $a^{(n)}\to\infty$ if $a^{(n)}_i\to\infty$ for any $i=1,\,\ldots,\,d$.

\begin{defn}[Trace distance]
The \emph{trace distance} on $\mathcal{S}_{\mathbb{Z}^d}$ is the distance induced by the norm on $\mathfrak{U}_{\mathbb{Z}^d}$: For any $\rho,\,\sigma\in\mathcal{S}_{\mathbb{Z}^d}$,
\begin{equation}
T(\rho,\sigma) = \frac{1}{2}\sup\left\{\left|\rho(A) - \sigma(A)\right|:A\in\mathfrak{U}_{\mathbb{Z}^d},\,\left\|A\right\|_\infty\le1\right\}\,.
\end{equation}
\end{defn}

\begin{prop}\label{prop:TI}
The trace distance on $\mathcal{S}_{\mathbb{Z}^d}$ is the supremum of the trace distances between the marginal states: For any $\rho,\,\sigma\in\mathcal{S}_{\mathbb{Z}^d}$,
\begin{equation}
T(\rho,\sigma) = \frac{1}{2}\sup_{\Lambda\in\mathcal{F}_{\mathbb{Z}^d}}\left\|\rho_\Lambda - \sigma_\Lambda\right\|_1\,,
\end{equation}
where $\|\cdot\|_1$ denotes the trace norm on $\mathfrak{U}_\Lambda$ given by
\begin{equation}
\left\|A\right\|_1 = \mathrm{Tr}_\Lambda\sqrt{A^\dag A}\,,\qquad A\in\mathfrak{U}_\Lambda\,.
\end{equation}
\end{prop}

\begin{proof}
See \autoref{sec:propTI}.
\end{proof}

\begin{defn}[Specific entropy {\cite[Proposition 6.2.38]{bratteli2013operatorII}}]
The \emph{von Neumann entropy} of a quantum state $\rho$ acting on a finite-dimensional Hilbert space is \cite{nielsen2010quantum,wilde2017quantum,holevo2019quantum}
\begin{equation}
S(\rho) = - \mathrm{Tr}\left[\rho\ln\rho\right]\,.
\end{equation}
The \emph{specific entropy} of $\rho\in\mathcal{S}_{\mathbb{Z}^d}^I$ is the entropy per site in the limit of infinite volume:
\begin{equation}
s(\rho) = \lim_{a\to\infty}\frac{S(\rho_{\Lambda_a})}{\left|\Lambda_a\right|} = \inf_{a\in\mathbb{N}_+^d}\frac{S(\rho_{\Lambda_a})}{|\Lambda_a|}\,.
\end{equation}
\end{defn}

\begin{defn}[Specific relative entropy \cite{jaksic2022approach}]
The \emph{relative entropy} \cite{nielsen2010quantum,wilde2017quantum,holevo2019quantum} between the quantum states $\rho$ and $\sigma$ acting on a finite-dimensional Hilbert space is
\begin{equation}
S(\rho\|\sigma) = -\mathrm{Tr}\left[\rho\left(\ln\rho - \ln\sigma\right)\right]\,.
\end{equation}
The \emph{specific relative entropy} between the states $\rho,\,\sigma\in\mathcal{S}_{\mathbb{Z}^d}^I$ is the relative entropy per site in the limit of infinite volume:
\begin{equation}
s(\rho\|\sigma) = \lim_{a\to\infty}\frac{S(\rho_{\Lambda_a}\|\sigma_{\Lambda_a})}{\left|\Lambda_a\right|}\,,
\end{equation}
whenever the limit exists.
\end{defn}

\begin{rem}
$s(\rho\|\sigma)=0$ does not imply $\rho=\sigma$.
Indeed, let $\rho_\Lambda = |0\rangle\langle0|^{\otimes\Lambda}$ and $\sigma_\Lambda = \frac{1}{2}\left(|0\rangle\langle0|^{\otimes\Lambda} + |1\rangle\langle1|^{\otimes\Lambda}\right)$ for any $\Lambda\in\mathcal{F}_{\mathbb{Z}^d}$.
Then, $S(\rho_\Lambda\|\sigma_\Lambda) =\ln2$ and $s(\rho\|\sigma)=0$.
\end{rem}

\subsection{Interactions}

An interaction $\Phi$ is a collection of observables $\left(\Phi(\Lambda)\right)_{\Lambda\in\mathcal{F}_{\mathbb{Z}^d}}$, where $\Phi(\Lambda)\in\mathcal{O}_\Lambda$ for any $\Lambda\in\mathcal{F}_{\mathbb{Z}^d}$.
We can associate to $\Phi$ the formal Hamiltonian
\begin{equation}\label{eq:HPhi}
H^\Phi_{\mathbb{Z}^d} = \sum_{\Lambda\in\mathcal{F}_{\mathbb{Z}^d}}\Phi(\Lambda)\,.
\end{equation}
In general the series \eqref{eq:HPhi} does not converge not even weakly, and $H^\Phi_{\mathbb{Z}^d}$ cannot be defined as an element of $\mathfrak{U}_{\mathbb{Z}^d}$.

We can define for any $\Lambda\in\mathcal{F}_{\mathbb{Z}^d}$ the local Hamiltonian on $\Lambda$ with open boundary conditions
\begin{equation}
H_\Lambda^\Phi = \sum_{X\subseteq\Lambda}\Phi(X) \in \mathcal{O}_\Lambda\,.
\end{equation}
An interaction $\Phi$ is translation invariant if $\Phi(\Lambda+x) = \tau_x(\Phi(\Lambda))$ for any $\Lambda\in\mathcal{F}_{\mathbb{Z}^d}$ and any $x\in\mathbb{Z}^d$.
For $r>0$, we denote with $\mathcal{B}_{\mathbb{Z}^d}^r$ the set of translation-invariant interactions satisfying
\begin{equation}\label{eq:defBr}
\|\Phi\|_r = \sum_{0\in\Lambda\in\mathcal{F}_{\mathbb{Z}^d}}e^{r\left(\left|\Lambda\right|-1\right)}\left\|\Phi(\Lambda)\right\|_\infty < \infty\,.
\end{equation}
The specific energy observable of $\Phi\in\mathcal{B}_{\mathbb{Z}^d}^r$ is
\begin{equation}
E_\Phi = \sum_{0\in\Lambda\in\mathcal{F}_{\mathbb{Z}^d}}\frac{\Phi(\Lambda)}{\left|\Lambda\right|} \in \mathcal{O}_{\mathbb{Z}^d}
\end{equation}
and satisfies \cite[Proposition 6.2.39]{bratteli2013operatorII}
\begin{equation}\label{eq:EPhiN}
\lim_{a\to\infty}\frac{1}{\left|\Lambda_a\right|}\left\|H^\Phi_{\Lambda_a} - \sum_{x\in\Lambda_a}\tau_x E_\Phi\right\|_\infty = 0\,.
\end{equation}
Therefore, for any $\rho\in\mathcal{S}_{\mathbb{Z}^d}^I$ we have that $\rho(E_\Phi)$ is equal to the average energy per site of $\rho$ in the limit of infinite volume:
\begin{equation}\label{eq:EPhi}
\lim_{a\to\infty}\frac{\rho\left(H^\Phi_{\Lambda_a}\right)}{\left|\Lambda_a\right|} = \rho\left(E_\Phi\right)\,.
\end{equation}

\subsection{Gibbs states}

Let $\Phi\in\mathcal{B}^r_{\mathbb{Z}^d}$.
For any $\Lambda\in\mathcal{F}_{\mathbb{Z}^d}$, the local Gibbs state of $\Phi$ on $\Lambda$ with open boundary conditions is the Gibbs state of $H^\Phi_\Lambda$:
\begin{equation}\label{eq:Gloc}
\omega^\Phi_\Lambda = \frac{e^{-H^\Phi_\Lambda}}{\mathrm{Tr}_\Lambda e^{-H^\Phi_\Lambda}} \in \mathcal{S}_\Lambda\,.
\end{equation}
Since the temperature can always be reabsorbed in the interaction, in the whole paper we set it to one.
\begin{rem}
The states $\left(\omega^\Phi_\Lambda\right)_{\Lambda\in\mathcal{F}_{\mathbb{Z}^d}}$ defined in \eqref{eq:Gloc} are not necessarily the marginal states of a single global state $\omega\in\mathcal{S}_{\mathbb{Z}^d}$.
\end{rem}

The pressure of $\Phi\in\mathcal{B}_{\mathbb{Z}^d}^r$ is \cite[Theorem 6.2.40]{bratteli2013operatorII}
\begin{equation}\label{eq:defP}
P(\Phi) = \lim_{a\to\infty}\frac{\ln\mathrm{Tr}_{\Lambda_a} e^{-H^\Phi_{\Lambda_a}}}{\left|\Lambda_a\right|} = \sup_{\rho\in\mathcal{S}^I_{\mathbb{Z}^d}}\left(s(\rho) - \rho(E_\Phi)\right)\,.
\end{equation}
The states $\rho\in\mathcal{S}^I_{\mathbb{Z}^d}$ that achieve the supremum in \eqref{eq:defP} are called equilibrium states of $\Phi$ and satisfy the Kubo--Martin--Schwinger condition \cite{bratteli2013operatorII}.
We denote with $\mathcal{S}_{eq}(\Phi)$ the set of such states.
For any $\Phi\in\mathcal{B}^r_{\mathbb{Z}^d}$, $\mathcal{S}_{eq}(\Phi)$ is nonempty, convex and compact with respect to the trace distance.

\section{The quantum \texorpdfstring{$W_1$}{W\_1} distance}\label{sec:W1}
Ref. \cite{de2021quantum} introduced the following generalization of the Wasserstein distance of order $1$ and of the Lipschitz constant to quantum systems made by a finite number of spins.
Since Ref. \cite{de2021quantum} does not require the lattice structure of $\mathbb{Z}^d$, here $\Lambda$ denotes a fixed generic finite set.
The quantum $W_1$ distance is based on the notion of neighboring quantum states.
The states $\rho,\,\sigma\in\mathcal{S}_\Lambda$ are neighboring if there exists $x\in\Lambda$ such that $\mathrm{Tr}_x\rho = \mathrm{Tr}_x\sigma$.
The quantum $W_1$ norm is the maximum norm that assigns distance at most $1$ to each couple of neighboring states:
\begin{defn}[$W_1$ norm]\label{defn:W1n}
Let $\Lambda$ be a finite set.
We define for any $\Delta\in\mathcal{O}_\Lambda^T$
\begin{equation}
\left\|\Delta\right\|_{W_1} = \frac{1}{2}\min\left\{\sum_{x\in\Lambda}\left\|\Delta^{(x)}\right\|_1 : \Delta^{(x)}\in\mathcal{O}^T_\Lambda\,,\;\mathrm{Tr}_x \Delta^{(x)} = 0\,,\; \sum_{x\in\Lambda}\Delta^{(x)} = \Delta\right\}\,.
\end{equation}
\end{defn}
The \emph{quantum $W_1$ distance} on $\mathcal{S}_\Lambda$ is the distance induced by the quantum $W_1$ norm: For any $\rho,\,\sigma\in\mathcal{S}_\Lambda$,
\begin{equation}
W_1(\rho,\sigma) = \left\|\rho - \sigma\right\|_{W_1}\,.
\end{equation}

\begin{defn}[Lipschitz constant]\label{defn:partial}
Let $\Lambda$ be a finite set.
We define the dependence of $H\in\mathcal{O}_\Lambda$ on the site $x\in\Lambda$ as
\begin{equation}\label{eq:partialxH}
\partial_x H = 2\min_{H_{\Lambda\setminus x}\in\mathcal{O}_{\Lambda\setminus x}}\left\|H - H_{\Lambda\setminus x}\right\|_\infty\,.
\end{equation}
The \emph{quantum Lipschitz constant} of $H\in\mathcal{O}_\Lambda$ is
\begin{equation}
\left\|H\right\|_L = \max_{x\in\Lambda}\partial_xH\,.
\end{equation}
\end{defn}

\begin{prop}[{\cite[Proposition 8]{de2021quantum}}]\label{prop:duality}
The quantum $W_1$ norm and the quantum Lipschitz constant are mutually dual, \emph{i.e.}, for any $\Delta\in\mathcal{O}_\Lambda^T$ we have
\begin{equation}
\left\|\Delta\right\|_{W_1} = \max\left\{\mathrm{Tr}_\Lambda\left[\Delta\,H\right]:H\in\mathcal{O}_\Lambda,\,\left\|H\right\|_L\le1\right\}\,.
\end{equation}
\end{prop}

\section{The quantum \texorpdfstring{$W_1$}{W\_1} distance for infinite lattices}\label{sec:w1}
In this section we extend the quantum Wasserstein distance of order $1$ of Ref. \cite{de2021quantum} to the quantum states of the spin lattice $\mathbb{Z}^d$.
As for the entropy and the relative entropy, we define a specific quantum $W_1$ distance, which we denote with $w_1$, as the $W_1$ distance per site in the limit of infinite volume:
\begin{defn}[Specific quantum $W_1$ distance]\label{def:w1}
For any $\rho,\,\sigma\in\mathcal{S}_{\mathbb{Z}^d}^I$ we define
\begin{equation}\label{eq:w1def}
w_1(\rho,\sigma) = \lim_{a\to\infty}\frac{\left\|\rho_{\Lambda_a} - \sigma_{\Lambda_a}\right\|_{W_1}}{\left|\Lambda_a\right|}\,.
\end{equation}
\end{defn}

\begin{rem}
We define the specific quantum $W_1$ distance only for translation-invariant states since the limit \eqref{eq:w1def} may not exist for generic states in $\mathcal{S}_{\mathbb{Z}^d}$.
\end{rem}

\begin{prop}\label{prop:w1sup}
The limit in \eqref{eq:w1def} exists for any $\rho,\,\sigma\in\mathcal{S}_{\mathbb{Z}^d}^I$ and is equal to
\begin{equation}\label{eq:w1sup}
w_1(\rho,\sigma) = \sup_{a\in\mathbb{N}_+^d}\frac{\left\|\rho_{\Lambda_a} - \sigma_{\Lambda_a}\right\|_{W_1}}{|\Lambda_a|}\,.
\end{equation}
Moreover, $w_1$ is a distance on $\mathcal{S}_{\mathbb{Z}^d}^I$.
\end{prop}

\begin{proof}
For any $a\in\mathbb{N}_+^d$, let
\begin{equation}
f(a) = \left\|\rho_{\Lambda_a} - \sigma_{\Lambda_a}\right\|_{W_1}\,.
\end{equation}
Given $a\in\mathbb{N}_+^d$, $k\in\mathbb{N}$ and $i\in\left\{1,\,\ldots,\,d\right\}$, let
\begin{equation}
b=\left(a_1,\,\ldots,\,a_i + k,\,\ldots,\,a_d\right)\,,\qquad c=\left(a_1,\,\ldots,\,k,\,\ldots,\,a_d\right)\,.
\end{equation}
We have
\begin{equation}
\Lambda_b = \left(\Lambda_a - k\,e_i\right)\cup\left(\Lambda_c + a_i\,e_i\right)\,,
\end{equation}
where $e_i$ is the $i$-th vector of the canonical basis of $\mathbb{R}^d$.
Then, we get from \autoref{prop:W1SA} and from the translation invariance of $\rho$ and $\sigma$
\begin{equation}
f(b) \ge \left\|\rho_{\Lambda_a - k e_i} - \sigma_{\Lambda_a - k e_i}\right\|_{W_1} + \left\|\rho_{\Lambda_c + a_i e_i} - \sigma_{\Lambda_c + a_i e_i}\right\|_{W_1} = f(a) + f(c)\,.
\end{equation}
Then, $f$ is superadditive in each variable, and we have from the multidimensional Fekete's lemma \autoref{lem:Fekete}
\begin{equation}
w_1(\rho,\sigma) = \lim_{a\to\infty}\frac{f(a)}{\left|\Lambda_a\right|} = \sup_{a\in\mathbb{N}_+^d}\frac{f(a)}{\left|\Lambda_a\right|}\,.
\end{equation}

The nonnegativity and the triangle inequality for $w_1$ follow from the nonnegativity and the triangle inequality for $W_1$, respectively.
Let $w_1(\rho,\sigma) = 0$.
Then, \eqref{eq:w1sup} implies
\begin{equation}
\left\|\rho_{\Lambda_a} - \sigma_{\Lambda_a}\right\|_{W_1} = 0 \qquad\forall\,a\in\mathbb{N}_+^d\,,
\end{equation}
\emph{i.e.}, $\rho_{\Lambda_a} = \sigma_{\Lambda_a}$ for any $a\in\mathbb{N}_+^d$.
Let $\Lambda\in\mathcal{F}_{\mathbb{Z}^d}$, and let us choose $a\in\mathbb{N}_+^d$ such that $\Lambda\subseteq\Lambda_a$.
Then,
\begin{equation}
\rho_\Lambda = \mathrm{Tr}_{\Lambda_a\setminus\Lambda}\rho_{\Lambda_a} = \mathrm{Tr}_{\Lambda_a\setminus\Lambda}\sigma_{\Lambda_a} = \sigma_\Lambda\,,
\end{equation}
hence $\rho = \sigma$.
\end{proof}

The specific quantum $W_1$ distance is always upper bounded by the trace distance:
\begin{prop}
For any $\rho,\,\sigma\in\mathcal{S}_{\mathbb{Z}^d}^I$ we have
\begin{equation}
w_1(\rho,\sigma) \le T(\rho,\sigma)\,.
\end{equation}
Moreover, for any $a\in\mathbb{N}_+^d$ we have
\begin{equation}
\left\|\rho_{\Lambda_a} - \sigma_{\Lambda_a}\right\|_1 \le 2\left|\Lambda_a\right|w_1(\rho,\sigma)\,.
\end{equation}
\end{prop}

\begin{proof}
We have
\begin{equation}
w_1(\rho,\sigma) = \sup_{a\in\mathbb{N}_+^d}\frac{\left\|\rho_{\Lambda_a} - \sigma_{\Lambda_a}\right\|_{W_1}}{|\Lambda_a|} \overset{\mathrm{(a)}}{\le} \sup_{a\in\mathbb{N}_+^d}\frac{\left\|\rho_{\Lambda_a} - \sigma_{\Lambda_a}\right\|_1}{2} \overset{\mathrm{(b)}}{\le} T(\rho,\sigma)\,,
\end{equation}
where (a) follows from \autoref{prop:W1T} and (b) follows from \autoref{prop:TI}.
From \autoref{prop:W1T}, we have for any $a\in\mathbb{N}_+^d$
\begin{equation}
w_1(\rho,\sigma) \ge \frac{\left\|\rho_{\Lambda_a} - \sigma_{\Lambda_a}\right\|_{W_1}}{|\Lambda_a|} \ge \frac{\left\|\rho_{\Lambda_a} - \sigma_{\Lambda_a}\right\|_{1}}{2\left|\Lambda_a\right|}\,.
\end{equation}
The claim follows.
\end{proof}

The specific quantum $W_1$ distance admits an equivalent definition, which directly generalizes \autoref{defn:W1n} to infinite lattices and does not require the limit over hypercubes.

\begin{defn}\label{def:alternative-wass}
We define for any $\sigma$, $\rho \in \mathcal{S}_{\mathbb{Z}^d}^I$,
\begin{align}\label{eq:definition-w1-distance}
\left\| \rho - \sigma \right\|_{w_1} =  \inf\Bigg\{ c\ge0 : \exists \, \rho', \, \sigma' \in \mathcal{S}_{\mathbb{Z}^d} :\; &\rho'_{\mathbb{Z}^d \setminus 0} = \sigma '_{\mathbb{Z}^d\setminus 0}\,,\nonumber\\
&\rho_\Lambda - \sigma_\Lambda = c \sum_{x \in \Lambda } (\tau_x\rho')_{\Lambda} - (\tau_x\sigma')_{\Lambda}\quad\forall \, \Lambda \in \mathcal{F}_{\mathbb{Z}^d}\Bigg\}\,.
 \end{align}
 \end{defn}
Let us collect some basic properties of the quantity defined above.

\begin{prop}
The infimum in \eqref{eq:definition-w1-distance} is attained for some $c \ge 0$, $\rho', \sigma'\in \mathcal{S}_{\mathbb{Z}^d}$. Moreover, given sequences $(\rho_n)_{n \in \mathbb{N}}$, $(\sigma_n)_{n \in \mathbb{N}} \subseteq \mathcal{S}_{\mathbb{Z}^d}^I$ weakly converging respectively towards $\rho$ and $\sigma$, then
\begin{equation}\label{eq:lsc}
\| \rho - \sigma\|_{w_1} \le \liminf_{ n \to \infty} \| \rho_n - \sigma_n\|_{w_1}\, .
\end{equation}
\end{prop}

\begin{proof}
Both statements follow from the weak sequential compactness of $\mathcal{S}_{\mathbb{Z}^d}$, together with the fact that for every $\Lambda \in \mathcal{F}_{\mathbb{Z}^d}$, the restriction map on states $\rho \mapsto \rho_{\Lambda}$ is weakly continuous. Considering a sequence $(c_n,\rho'_n,\sigma_n')_{n \in \mathbb{N}}$ such that $\lim_n c_n = \| \rho - \sigma\|_{w_1}$ and, using compactness to extract converging subsequences, assume that $\lim_n \rho_n' = \rho'$, $\lim_n \sigma_n' = \sigma'$. By continuity we deduce that $c$, $\rho'$, $\sigma'$ satisfy the conditions in \eqref{eq:definition-w1-distance}, hence they are minimizers. A similar argument gives \eqref{eq:lsc}.
\end{proof}

The rest of this section is devoted to showing the equivalence between \autoref{def:w1} and \autoref{def:alternative-wass}.

 \begin{thm}\label{thm:alternative-w1}
 For  $\sigma$, $\rho \in \mathcal{S}_{\mathbb{Z}^d}^I$,  we have
 \begin{equation} w_1(\rho, \sigma) = \| \rho- \sigma\|_{w_1}\, .\end{equation}
 \end{thm}

We split the proof into several intermediate results. We begin with the following upper bound.

\begin{lem}[Upper bound]\label{prop:le}
For  $\sigma$, $\rho \in \mathcal{S}_{\mathbb{Z}^d}^I$,  we have
\begin{equation}\label{eq:upper-bound} w_1(\rho, \sigma)  \le \| \rho - \sigma\|_{w_1} \, .\end{equation}
\end{lem}

\begin{proof}
Given $c\ge 0$, $\rho'$, $\sigma' \in \mathcal{S}_{\mathbb{Z}^d}$ as in the right hand side of \eqref{eq:definition-w1-distance},
for any $a \in \mathbb{N}^d_+$, we write the identity
\begin{equation} \rho_{\Lambda_a} - \sigma_{\Lambda_a} =   \sum_{x \in \Lambda_a} c \left( \rho^{(x)}_{\Lambda_a} - \sigma^{(x)}_{\Lambda_a}\right)\, ,\end{equation}
where we define, for $x \in \Lambda_a$, the states $\rho^{(x)}_{\Lambda_a} = (\tau_x \rho')_{\Lambda_a}$ and $\sigma^{(x)}_{\Lambda_a} = (\tau_x \sigma')_{\Lambda_a}$. Using \eqref{eq:marginal-marginal}, it follows that
\begin{equation}
\mathrm{Tr}_{x} \rho^{(x)}_{\Lambda_a} = \mathrm{Tr}_{x} \left[ (\tau_x \rho')_{\Lambda_a} \right]= (\tau_x \rho')_{\Lambda_a \setminus x} = \tau_x\rho'_{(\Lambda_a-x) \setminus 0} \, ,
\end{equation}
and similarly with $\sigma^{(x)}_{\Lambda_a}$, so that  $\mathrm{Tr}_{x} \rho^{(x)}_{\Lambda_a}  = \mathrm{Tr}_{x} \sigma^{(x)}_{\Lambda_a}$ for every $x \in \Lambda_a$. Therefore, by definition of $W_1$ on $\mathcal{S}_{\Lambda_a}$, we have the inequality
\begin{equation} \| \rho_{\Lambda_a} - \sigma_{\Lambda_a} \|_{W_1} \le c | \Lambda_a | \, .\end{equation}
Dividing both sides by $|\Lambda_a|$ and letting $a \to \infty$, we deduce $w_1(\rho, \sigma)  \le c$, hence \eqref{eq:upper-bound}.
\end{proof}

To establish the lower bound, we consider a periodic approximation of the marginal states over a box $\Lambda_a$. We write, for any $a \in \mathbb{N}_+^d$ and $k \in \mathbb{Z}^d$,
\begin{equation} 2ak =(2a_ik_i)_{i=1}^d\, .
\end{equation}

\begin{prop}[Periodic approximation]\label{prop:periodic}
For  $\sigma$, $\rho \in \mathcal{S}_{\mathbb{Z}^d}^I$,  and $a \in \mathbb{N}_+^d$, define  $\tilde{\rho}^a$, $\tilde{\sigma}^a \in \mathcal{S}_{\mathbb{Z}^d}^I$ as follows:
\begin{equation}
\tilde \rho^a = \frac{1}{|\Lambda_a|} \sum_{x \in \Lambda_a} \tau_x \bigotimes_{k \in \mathbb{Z}^d} \tau_{2ak} \rho_{\Lambda_a}\, , \qquad \tilde \sigma^a = \frac{1}{|\Lambda_a|} \sum_{x \in \Lambda_a} \tau_x \bigotimes_{k \in \mathbb{Z}^d} \tau_{2ak} \sigma_{\Lambda_a}\, .
\end{equation}
Then, we have
\begin{equation}\label{eq:inequality-periodic-w1}
\| \tilde \rho^{a} - \tilde \sigma^{a} \|_{w_1} \le  \frac { \| \rho_{\Lambda_a} - \sigma_{\Lambda_a} \|_{W_1}}{|\Lambda_a|}  \, .
\end{equation}
\end{prop}

\begin{proof}
We introduce first some notation. For disjoint sets $R, S \subseteq \mathbb{Z}^d$, write
\begin{equation}
\alpha_{R, S} = \left(\bigotimes_{k \in R} \tau_{2ak} \rho_{\Lambda_a}\right) \otimes \left(\bigotimes_{k \in S}\tau_{2ak} \sigma_{\Lambda_a}\right)\, ,
\end{equation}
which we further simplify to $\alpha_R = \alpha_{R,S}$ whenever $S = \mathbb{Z}^d \setminus R$.
With the above notation, we have
\begin{equation}\label{eq:tilde-rho-a}
\tilde{\rho}^a = \frac{1}{|\Lambda_a|} \sum_{x \in \Lambda_a} \tau_x \alpha_{\mathbb{Z}^d}\, , \qquad \tilde{\sigma}^a = \frac{1}{|\Lambda_a|} \sum_{x \in \Lambda_a} \tau_x \alpha_{\emptyset}\, .
\end{equation}
For $x \in \Lambda_a$, let $c_{x} \ge 0$ and $\rho^{(x)}_{\Lambda_a}, \sigma^{(x)}_{\Lambda_a} \in \mathcal{S}_{\Lambda_a}$ be such that
\begin{equation}\label{eq:rho-x-sigma-x}
\mathrm{Tr}_x \rho^{(x)}_{\Lambda_a} = \mathrm{Tr}_x \sigma^{(x)}_{\Lambda_a}\; , \qquad \rho_{\Lambda_a} - \sigma_{\Lambda_a} = \sum_{x \in \Lambda_a} c_{x} \left( \rho^{(x)}_{\Lambda_a} - \sigma^{(x)}_{\Lambda_a}\right)\,,
\end{equation}
and introduce the states
\begin{equation}   \tilde\rho^{(x)}  =  \alpha_{R_d, S_d}  \otimes  \rho^{(x)}_{\Lambda_a} \, , \qquad
  \tilde\sigma^{(x)} = \alpha_{R_d, S_d}  \otimes  \sigma^{(x)}_{\Lambda_a}\, ,
  \end{equation}
  where $R_d, S_d$ are disjoint sets with $R_d \cup S_d =  \mathbb{Z}^d\setminus 0$, to be specified in \eqref{eq:recursive-rd} below (their precise definition will be  relevant only later). Notice that
  \begin{equation}\label{eq:identity-trace-x-rho-x}
  \mathrm{Tr}_x \tilde{\rho}^{(x)} = \alpha_{R_d, S_d}\otimes\mathrm{Tr}_x  \rho^{(x)}_{\Lambda_a} = \alpha_{R_d, S_d}\otimes\mathrm{Tr}_x  \sigma^{(x)}_{\Lambda_a} = \mathrm{Tr}_x \tilde{\sigma}^{(x)} \, ,
  \end{equation}
hence, for every $z \in \mathbb{Z}^d$, $\Lambda \in \mathcal{F}_{\mathbb{Z}^d}$ with $x+ z \notin \Lambda$,
  \begin{equation}\label{eq:identity-marginal-rho-x}
   \left( \tau_{z} \tilde{\rho}^{(x)}\right)_\Lambda = \left( \tau_{z} \tilde{\sigma}^{(x)}\right)_\Lambda\, ,
   \end{equation}
    Moreover,
   \begin{align}\label{eq:telescopic-one-term}
  \sum_{x \in \Lambda_a} c_x \left(  \tilde{\rho}^{(x)}  -  \tilde{\sigma}^{(x)} \right) & = \alpha_{R_d, S_d}\otimes \sum_{x \in \Lambda_a} c_x \left(  \rho^{(x)}_{\Lambda_a}  -  \sigma^{(x)}_{\Lambda_a} \right) = \alpha_{R_d ,S_d} \otimes \left(\rho_{\Lambda_a} - \sigma_{\Lambda_a}\right) \nonumber \\
   & = \alpha_{R_d \cup 0} -\alpha_{R_d}\, .
   \end{align}

We assume that $c = \sum_{x \in \Lambda_a} c_x >0$, otherwise \eqref{eq:rho-x-sigma-x}  yields $\rho_{\Lambda_a} = \sigma_{\Lambda_a}$ hence $\tilde{\rho}^{a} = \tilde{\sigma}^{a}$ and \eqref{eq:inequality-periodic-w1} holds since $\| \tilde\rho^a - \tilde\sigma^a\|_{w_1} = 0$, by choosing $c=0$ and any $\rho' =\sigma' \in \mathcal{S}_{\mathbb{Z}^d}$. Therefore, letting $p_x = c_x/c$, we define the states
 \begin{equation}
 \rho' =  \sum_{x \in \Lambda} p_x \tau_{-x} \tilde{\rho}^{(x)}\, ,  \quad \sigma' =   \sum_{x \in \Lambda} p_x \tau_{-x} \tilde{\sigma}^{(x)}\, .
 \end{equation}
By \eqref{eq:identity-trace-x-rho-x}, we have the identity
 \begin{equation}
\mathrm{Tr}_0 \rho'  = \sum_{x \in \Lambda} p_x \mathrm{Tr}_0   \tau_{-x} \tilde{\rho}^{(x)}  =  \sum_{x \in \Lambda} p_x\tau_{-x} \mathrm{Tr}_x \tilde{\rho}^{(x)} =  \sum_{x \in \Lambda} p_x \tau_{-x}\mathrm{Tr}_x \tilde{\sigma}^{(x)}  = \mathrm{Tr}_0 \sigma'\, .
\end{equation}
If we prove that, for every $\Lambda \in \mathcal{F}_{\mathbb{Z}^d}$,
\begin{equation}\label{eq:sum-periodic-lambda}
\tilde{\rho}^a_{\Lambda} - \tilde{\sigma}^a_{\Lambda} = c \sum_{y \in \Lambda}\left( (\tau_y \rho')_{\Lambda} - (\tau_y \sigma')_{\Lambda}\right)\, ,
\end{equation}
then
\begin{equation}
\| \tilde{\rho}^a - \tilde{\sigma}^a\|_{w_1} \le c\, ,
\end{equation}
and \eqref{eq:inequality-periodic-w1} follows.
To show \eqref{eq:sum-periodic-lambda}, we write explicitly
\begin{align}
 c \sum_{y \in \Lambda} \left((\tau_{y} \rho')_{\Lambda} -  (\tau_{y} \sigma')_{\Lambda}\right) &  =  \sum_{y \in \Lambda } \sum_{x \in \Lambda_a} c_x \left( ( \tau_{y-x} \tilde{\rho}^{(x)})_{\Lambda} -  ( \tau_{y-x} \tilde{\sigma}^{(x)})_{\Lambda} \right) \nonumber \\ & = \sum_{x \in \Lambda_a} c_x \sum_{y \in \Lambda}  ( \tau_{y-x} \tilde{\rho}^{(x)})_{\Lambda} -  ( \tau_{y-x} \tilde{\sigma}^{(x)})_{\Lambda}  \nonumber \\
 & = \sum_{x \in \Lambda_a} c_x \sum_{z \in \Lambda -x }   ( \tau_{z} \tilde{\rho}^{(x)})_{\Lambda} -  ( \tau_{z} \tilde{\sigma}^{(x)})_{\Lambda} \,  .
\end{align}
where the last line follows letting $z = y-x$. Using \eqref{eq:identity-marginal-rho-x}, we extend the summation  over $z \in \Lambda'$, for any $\Lambda' \in \mathcal{F}_{\mathbb{Z}^d}$ such that
\begin{equation}\label{eq:lambda-prime-large}
\Lambda'\supseteq \bigcup_{x \in \Lambda_a} ( \Lambda-x)\, .
\end{equation}
 Then, exchanging again the order of summation and using \eqref{eq:telescopic-one-term},
\begin{align}
\sum_{x \in \Lambda_a} c_x \sum_{z \in  \Lambda' }  ( \tau_{z} \tilde{\rho}^{(x)} )_{\Lambda} -  ( \tau_{z} \tilde{\sigma}^{(x)})_{\Lambda}   & = \sum_{z \in  \Lambda' }   \sum_{x \in \Lambda_a} c_x  \left( ( \tau_{z} \tilde{\rho}^{(x)} )_{\Lambda} -  ( \tau_{z} \tilde{\sigma}^{(x)})_{\Lambda}) \right) \\
& = \sum_{z \in  \Lambda' }  (\tau_z \alpha_{R_d \cup 0})_{\Lambda} - ( \tau_z \alpha_{R_d})_{\Lambda}
\end{align}
We now specify the sets $\Lambda'$ and $R_d$  in such a way that the above summation is telescopic and yields \eqref{eq:sum-periodic-lambda}. First, we let
\begin{equation}
\Lambda' = \bigcup_{k \in \Lambda_b} \left( \Lambda_a + 2ak\right)\, ,
\end{equation}
with $b \in \mathbb{N}^d_+$ sufficiently large so that \eqref{eq:lambda-prime-large} holds. Then, recalling \eqref{eq:tilde-rho-a}, to obtain \eqref{eq:sum-periodic-lambda} it is sufficient to prove that
\begin{equation}\label{eq:telescopic}
\sum_{k \in  \Lambda_b }  (\tau_{2ak} \alpha_{R_d \cup 0})_{\Lambda} - ( \tau_{2ak} \alpha_{R_d})_{\Lambda} = (\alpha_{\mathbb{Z}^d})_{\Lambda} - (\alpha_\emptyset)_{\Lambda}\, .
\end{equation}
The following recursive definition for the subsets $R_d \subseteq \mathbb{Z}^d\setminus0$ serves exactly this purpose. We let
\begin{equation}\label{eq:recursive-rd}
R_1 = \mathbb{Z}_-\, , \; R_{d} = \left( \mathbb{Z}^{d-1} \times \mathbb{Z}_-\right) \cup \left(R_{d-1} \times 0 \right)\, ,
\end{equation}
so that
\begin{equation}
R_d \cup \{ 0 \} = R_d+ e_1\, ,
\end{equation}
where we write $e_i \in \mathbb{Z}^d$ for the natural basis vectors, for $i=1, \ldots, d$ (see \autoref{fig:R_d}).

\begin{figure}[ht]
\begin{minipage}{0.45\textwidth}
\begin{center}
\begin{tikzpicture}[scale = 1.2, every node/.style={minimum size=13, outer sep=0pt}]
	\draw[step=0.5, color=black!10, style=dashed] (0.1,0.1) grid (3.4,3.4);

    \foreach \y in {0.25,0.75,1.25} {
        \foreach \x in {0.25, 0.75, 1.25, 1.75, 2.25, 2.75, 3.25} {
   			 \node[fill=black!30, rectangle] at (\x,\y) {};
   			 }
    }

    \foreach \x in {0.25, 0.75, 1.25} {
   			 \node[fill=black!30] at (\x,1.75) {};
    }

\draw[->, color=black] (1.75, 0) -- (1.75, 3.4);
\draw[->, color=black] (0,1.75) -- (3.4,1.75);
\end{tikzpicture}
\end{center}
\end{minipage}
\begin{minipage}{0.45\textwidth}
\begin{center}
\begin{tikzpicture}[scale = 1.2, every node/.style={minimum size=13, outer sep=0pt}]
	\draw[step=1,color=black!10, style=dashed] (0.1,0.1) grid (3.4,3.4);

    \foreach \y in {0.25,0.75,1.25} {
        \foreach \x in {0.25, 0.75, 1.25, 1.75, 2.25, 2.75, 3.25} {
   			 \node[fill=black!30] at (\x,\y) {};
   			 }
    }

    \foreach \x in {0.25, 0.75, 1.25, 1.75} {
   			 \node[fill=black!30] at (\x,1.75) {};
    }

\draw[->, color=black] (1.75, 0) -- (1.75, 3.4);
\draw[->, color=black] (0,1.75) -- (3.4,1.75);
\end{tikzpicture}
\end{center}
\end{minipage}
\caption{representation of $R_2 \subseteq \mathbb{Z}^2$ (left) and its translated $R_2+e_1 = R_2 \cup 0$ (right). \label{fig:R_d}}
\end{figure}
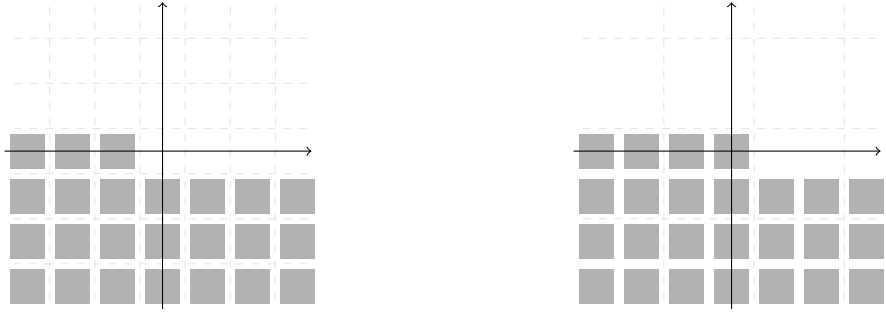

We decompose  the right hand side in \eqref{eq:telescopic} as a double summation, over $k_1$ and $k_{\setminus 1} = (-b_1, k_2, \ldots, k_{d})$, so that, for fixed $k_{\setminus 1}$, we find a telescopic sum
\begin{align}
\sum_{0 \le  k_1 < 2 b_1} \tau_{2ak_{\setminus 1}} \alpha_{R_d + e_1}  -  \tau_{2ak_{\setminus 1}} \alpha_{R_d} & = \sum_{0 \le  k_1 < 2b_1} \tau_{2ak_{\setminus 1}} \alpha_{R_d + (k_1+1)e_1 } -  \tau_{2ak_{\setminus 1}} \alpha_{R_d+k_1 e_1 } \nonumber\\
& = \tau_{2ak_{\setminus 1}} \alpha_{R_d + 2b_1 e_1 } - \tau_{2ak_{\setminus 1}} \alpha_{R_d}\, .
\end{align}
Since $\Lambda \subseteq \Lambda'$, and
\begin{equation}\label{eq:identity-local-shift-above}
(R_d + 2b_1 e_1 + k_{\setminus 1}) \cap \Lambda_b = (R_d + e_{2} + k_{\setminus 1}) \cap \Lambda_b\, ,
\end{equation}
(see \autoref{fig:Lambda_b}), it follows that
\begin{equation}
(\tau_{2ak_{\setminus 1}} \alpha_{R_d + 2b_1 e_1})_{\Lambda}- (\tau_{2ak_{\setminus 1}} \alpha_{R_d})_{\Lambda} =  (\tau_{2a k_{\setminus 1}} \alpha_{R_d + e_{2}})_{\Lambda} - (\tau_{2ak_{\setminus 1}} \alpha_{R_d})_{\Lambda}\, .
\end{equation}

\begin{figure}[ht]
\begin{minipage}{0.45\textwidth}
\begin{center}
\begin{tikzpicture}[scale = 1.2, every node/.style={minimum size=13, outer sep=0pt}]
	\draw[step=0.5, color=black!10, style=dashed] (0.1,0.1) grid (3.4,3.4);

    \foreach \y in {0.25,0.75,1.25} {
        \foreach \x in {0.25, 0.75, 1.25, 1.75, 2.25, 2.75, 3.25} {
   			 \node[fill=black!30] at (\x,\y) {};
   			 }
    }

    \foreach \x in {0.25, 0.75, 1.25, 1.75, 2.25} {
   			 \node[fill=black!30] at (\x,1.75) {};
    }

\draw[->, color=black] (1.75, 0) -- (1.75, 3.4);
\draw[->, color=black] (0,1.75) -- (3.4,1.75);
\node at (2.25,2.25) {$\Lambda_b$};
\draw (0.5,0.5) -- (0.5,2.5);
\draw (0.5,0.5) -- (2.5,0.5);
\draw (2.5,2.5) -- (0.5,2.5);
\draw (2.5,2.5) -- (2.5,0.5);
\end{tikzpicture}
\end{center}
\end{minipage}
\begin{minipage}{0.45\textwidth}
\begin{center}
\begin{tikzpicture}[scale = 1.2, every node/.style={minimum size=13, outer sep=0pt}]
	\draw[step=1,color=black!10, style=dashed] (0.1,0.1) grid (3.4,3.4);

    \foreach \y in {0.25,0.75,1.25, 1.75} {
        \foreach \x in {0.25, 0.75, 1.25, 1.75, 2.25, 2.75, 3.25} {
   			 \node[fill=black!30] at (\x,\y) {};
   			 }
    }

    \foreach \x in {0.25} {
   			 \node[fill=black!30] at (\x,2.25) {};
    }

\draw[->, color=black] (1.75, 0) -- (1.75, 3.4);
\draw[->, color=black] (0,1.75) -- (3.4,1.75);
\node at (2.25,2.25) {$\Lambda_b$};
\draw (0.5,0.5) -- (0.5,2.5);
\draw (0.5,0.5) -- (2.5,0.5);
\draw (2.5,2.5) -- (0.5,2.5);
\draw (2.5,2.5) -- (2.5,0.5);
\end{tikzpicture}
\end{center}
\end{minipage}
\caption{the rectangle $\Lambda_b$ with $b=(2,2)$ is highlighted, providing an example of the general identity \eqref{eq:identity-local-shift-above}: in this case, $(R_2 + 2 e_1)\cap \Lambda_b$ (left) equals $(R_2 - 2e_1+e_2) \cap \Lambda_b$ (right). \label{fig:Lambda_b}}
\end{figure}
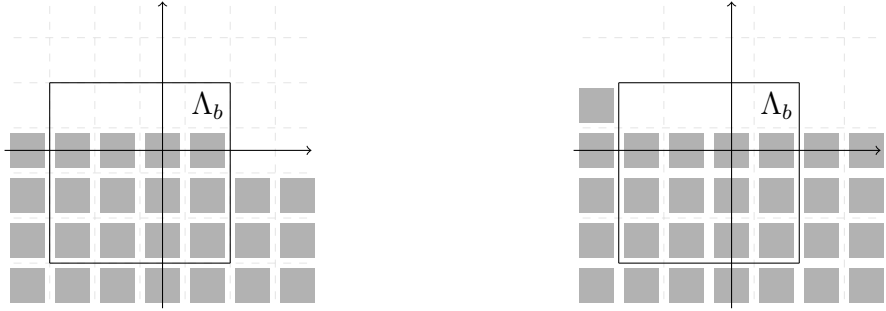
We further proceed decomposing the summation upon $k_2$ and
\begin{equation}
k_{\setminus 1,2} = (-b_1, -b_2, k_3, \ldots, k_d)\,,
\end{equation}
so that we obtain a similar telescopic sum. By iterating the same argument up to summation over $k_d$, we eventually conclude that
\begin{equation}
\sum_{k \in  \Lambda_b }  (\tau_{2ak} \alpha_{R_d \cup 0})_{\Lambda} - ( \tau_{2ak} \alpha_{R_d})_{\Lambda} = (\tau_{ - 2a b} \alpha_{R_d + 2b_d e_d})_{\Lambda} - (\tau_{ - 2a b} \alpha_{R_d})_{\Lambda}\, ,
\end{equation}
which gives \eqref{eq:telescopic} since
\begin{equation} (R_d+2b_d e_d - 2 ab) \cap \Lambda_b = \Lambda_b\, , \; (R_d - 2 ab) \cap \Lambda_b = \emptyset\, ,
\end{equation}
thus the proof  is completed.
\end{proof}

Using the above construction, we establish the following lower bound, hence completing the proof of \autoref{thm:alternative-w1}.

\begin{lem}[Lower bound]\label{prop:ge}
For  $\sigma$, $\rho \in \mathcal{S}_{\mathbb{Z}^d}^I$,  we have
\begin{equation}\label{eq:ge}
w_1(\rho, \sigma) \ge \| \rho - \sigma\|_{w_1}\, .
\end{equation}
\end{lem}

\begin{proof}
With the notation of \autoref{prop:periodic}, we argue that the states $\tilde{\rho}^a$ weakly converge to $\rho$. For any fixed $b \in \mathbb{N}^d_+$, if $x \in \Lambda_a$ is such that $\Lambda_b \subseteq \Lambda_a+x$, then
\begin{equation}
\left( \tau_x \bigotimes_{k \in \mathbb{Z}^d} \tau_{2ak} \rho_{\Lambda_a} \right)_{\Lambda_b} = \rho_{\Lambda_b}\, .
\end{equation}
Therefore, denoting by $G_a\subseteq \Lambda_a$ the set of such $x$'s, we write
\begin{equation}
\tilde{\rho}^a_{\Lambda_b} = \frac{|G_a|}{|\Lambda_a|}  \rho_{\Lambda_b} + \sum_{x \in \Lambda_a\setminus G_a}  \left( \tau_x \bigotimes_{k \in \mathbb{Z}^d} \tau_{2ak} \rho_{\Lambda_a} \right)_{\Lambda_b}
\end{equation}
Since $|G_a| = |\Lambda_a|- |\Lambda_b|$, it follows that, as $a \to \infty$, $\tilde{\rho}^a_{\Lambda_b}$ converge in $\mathcal{S}_{\Lambda_b}$ towards $\rho_{\Lambda_b}$. This holds for any $b \in \mathbb{N}^d_+$, hence we obtain the desired weak convergence in $\mathcal{S}_{\mathbb{Z}^d}$ of $\tilde{\rho}^a$ towards $\rho$. By \eqref{eq:lsc}, \autoref{prop:periodic} and \autoref{prop:w1sup} we have the inequalities
\begin{align}\label{eq:inequalities-periodic-approx}
\| \rho- \sigma\|_{w_1} &\le \liminf_{a \to \infty} \| \tilde{\rho}^a - \tilde{\sigma}^a \|_{w_1}  \nonumber \\
& \le  \limsup_{a \to \infty} \| \tilde{\rho}^a - \tilde{\sigma}^a \|_{w_1} \le \sup_{a \in \mathbb{N}^d_+} \frac{ \| \rho_{\Lambda_a} -  \sigma_{\Lambda_a} \|_{W_1}}{|\Lambda_a|} =   w_1(\rho, \sigma)\, ,
\end{align}
and the proof of \eqref{eq:ge} is completed.
\end{proof}

As a consequence of the above argument, we also obtain that the periodic approximations  always converge with respect to  the specific quantum $W_1$ distance.

\begin{cor} With the notation of \autoref{prop:periodic}, we have
\begin{equation}
\lim_{a \to \infty} w_1( \tilde{\rho}^a, \tilde{\sigma}^a) =  w_1(\rho, \sigma)\, .
\end{equation}
\end{cor}
\begin{proof}
In \eqref{eq:inequalities-periodic-approx} all inequalities must be equalities, hence the limit.
\end{proof}

\section{The quantum Lipschitz constant for infinite lattices}\label{sec:L}
In this section we extend the definition of quantum Lipschitz constant of Ref. \cite{de2021quantum} to interactions on the quantum spin lattice $\mathbb{Z}^d$.

The first step towards defining a Lipschitz constant for interactions is to extend to $\mathcal{O}_{\mathbb{Z}^d}$ the definition of dependence on a site:
\begin{defn}\label{defn:partialI}
For any $H\in\mathcal{O}_{\mathbb{Z}^d}$ and any $x\in\mathbb{Z}^d$ we define
\begin{equation}\label{eq:partialxHI}
\partial_x H = 2\inf_{A\in \mathcal{O}_{\mathbb{Z}^d\setminus x}}\left\|H - A\right\|_\infty\,.
\end{equation}
\end{defn}

\begin{prop}\label{prop:equivpartial}
For any $\Lambda\in\mathcal{F}_{\mathbb{Z}^d}$, any $H\in\mathcal{O}_\Lambda$ and any $x\in\Lambda$, \eqref{eq:partialxH} and \eqref{eq:partialxHI} are equivalent.
\end{prop}

\begin{proof}
See \autoref{sec:equivpartial}.
\end{proof}

We can now define the Lipschitz constant for interactions:
\begin{defn}[Lipschitz constant]\label{def:L}
We define the \emph{Lipschitz constant} of $\Phi\in\mathcal{B}^r_{\mathbb{Z}^d}$ as the dependence of the formal Hamiltonian $H^\Phi_{\mathbb{Z}^d}$ on the site $0$:
\begin{equation}
\left\|\Phi\right\|_L = \partial_0 \sum_{0\in \Lambda \in \mathcal{F}_{\mathbb{Z}^d}} \Phi(\Lambda)\,,
\end{equation}
where the series converges absolutely in the $\|\cdot\|_\infty$ norm.
\end{defn}

The Lipschitz constant of $\Phi$ is also equal to the dependence on a fixed site of the local Hamiltonian on a region in the limit of infinite volume:
\begin{prop}
For any $\Phi\in\mathcal{B}_{\mathbb{Z}^d}^r$ and any $x\in\mathbb{Z}^d$ we have
\begin{equation}\label{eq:defPhiL}
\lim_{a\to\infty}\partial_xH^\Phi_{\Lambda_a} = \left\|\Phi\right\|_L\,.
\end{equation}
\end{prop}

\begin{proof}
We have
\begin{align}
\left|\left\|\Phi\right\|_L - \partial_xH^\Phi_{\Lambda_a}\right| &\overset{\mathrm{(a)}}{=}
\left|\left\|\Phi\right\|_L - \partial_0 H^\Phi_{\Lambda_a-x}\right| =
\left|\left\|\Phi\right\|_L - \partial_0 \sum_{0\in\Lambda\subseteq\Lambda_a-x}\Phi(\Lambda)\right|\nonumber\\
&\overset{\mathrm{(b)}}{\le} \partial_0 \sum_{0\in \Lambda \in \mathcal{F}_{\mathbb{Z}^d}:\Lambda\not\subseteq\Lambda_a-x} \Phi(\Lambda) \le 2\sum_{0\in \Lambda \in \mathcal{F}_{\mathbb{Z}^d}:\Lambda\not\subseteq\Lambda_a-x}\left\|\Phi(\Lambda)\right\|_\infty\,,
\end{align}
where (a) follows from the translation invariance of $\Phi$ and (b) follows since $\partial_0$ is a seminorm.
Since
\begin{equation}
\sum_{0\in \Lambda \in \mathcal{F}_{\mathbb{Z}^d}}\left\|\Phi(\Lambda)\right\|_\infty \le \left\|\Phi\right\|_r < \infty\,,
\end{equation}
we have
\begin{equation}
\lim_{a\to\infty}\left|\left\|\Phi\right\|_L - \partial_xH^\Phi_{\Lambda_a}\right| \le 2\lim_{a\to\infty}\sum_{0\in \Lambda \in \mathcal{F}_{\mathbb{Z}^d}:\Lambda\not\subseteq\Lambda_a-x}\left\|\Phi(\Lambda)\right\|_\infty = 0\,.
\end{equation}
The claim follows.
\end{proof}

\subsection{Physical equivalence}
Different interactions may give rise to the same formal Hamiltonian.
Such interactions are called physically equivalent.
The concept of physical equivalence between interactions was formally introduced in \cite{griffiths1971strict,roos1974strict,israel2015convexity}.
We adopt the definition of \cite[Section 4.7]{ruelle2004thermodynamic} and \cite{jaksic2022approach,jaksic2022note}.
The reader can find more details in \cite[Section 2.4.6]{van1993regularity}.

\begin{defn}
The interaction $\Phi,\,\Psi\in\mathcal{B}^r_{\mathbb{Z}^d}$ are \emph{physically equivalent} if any of the following equivalent conditions holds:
\begin{enumerate}
\item The local Hamiltonians of $\Phi$ and $\Psi$ normalized by the number of sites differ only by a constant in the limit of infinite volume:
\begin{equation}
\lim_{a\to\infty}\left\|\frac{H^{\Phi}_\Lambda - H^{\Psi}_\Lambda}{\left|\Lambda_a\right|} - \omega(E_{\Phi-\Psi})\,\mathbb{I}\right\|_\infty = 0\,,
\end{equation}
where $\omega\in\mathcal{S}^I_{\mathbb{Z}^d}$ is the uniform distribution, \emph{i.e.}, $\omega_\Lambda = \frac{\mathbb{I}_\Lambda}{q^{|\Lambda|}}$ for any $\Lambda\in\mathcal{F}_{\mathbb{Z}^d}$.
\item $\Phi$ and $\Psi$ generate the same time evolution: For any $A\in\mathfrak{U}_{\mathbb{Z}^d}$ and any $t\in\mathbb{R}$ we have
\begin{equation}
\lim_{a\to\infty} \left\|e^{i H^\Phi_{\Lambda_a} t}\,A\,e^{-i H^\Phi_{\Lambda_a} t} - e^{i H^\Psi_{\Lambda_a} t}\,A\,e^{-i H^\Psi_{\Lambda_a} t}\right\|_\infty = 0\,.
\end{equation}
\item For any $A\in\mathfrak{U}^{loc}_{\mathbb{Z}^d}$ we have
\begin{equation}
\sum_{\Lambda\in\mathcal{F}_{\mathbb{Z}^d}} \left[\Phi(\Lambda) - \Psi(\Lambda),\,A\right] = 0\,.
\end{equation}
\end{enumerate}
\end{defn}

\begin{prop}
The interaction $\Phi\in\mathcal{B}^r_{\mathbb{Z}^d}$ is physically equivalent to the null interaction iff $\left\|\Phi\right\|_L = 0$.
\end{prop}

\begin{proof}
\begin{itemize}
\item Let $\Phi$ be physically equivalent to the null interaction.
Let
\begin{equation}\label{eq:K}
K = \sum_{0\in\Lambda\in\mathcal{F}_{\mathbb{Z}^d}}\Phi(\Lambda)\,,
\end{equation}
where the series converges absolutely in the $\|\cdot\|_\infty$ norm.
We have for any $A\in\mathfrak{U}_0$
\begin{equation}
0 = \sum_{\Lambda\in\mathcal{F}_{\mathbb{Z}^d}} \left[\Phi(\Lambda),\,A\right] = \sum_{0\in\Lambda\in\mathcal{F}_{\mathbb{Z}^d}}\left[\Phi(\Lambda),\,A\right] = \left[K,\,A\right]\,,
\end{equation}
therefore $K\in\mathcal{O}_{\mathbb{Z}^d}\setminus0$ and $\left\|\Phi\right\|_L = \partial_0 K = 0$.

\item Let $\left\|\Phi\right\|_L = 0$.
Let $K$ be as in \eqref{eq:K}.
We have
\begin{equation}
2\inf_{H\in\mathcal{O}_{\mathbb{Z}^d\setminus0}}\left\|K - H\right\|_\infty = \partial_0 K = \left\|\Phi\right\|_L = 0\,,
\end{equation}
therefore there exists a sequence
$\left(K^{(n)}\right)_{n\in\mathbb{N}}\subset\mathcal{O}_{\mathbb{Z}^d\setminus0}$ such that
\begin{equation}
\lim_{n\to\infty}\left\|K - K^{(n)}\right\|_\infty = 0\,.
\end{equation}
Then, $K\in\mathcal{O}_{\mathbb{Z}^d\setminus0}$, and for any $A_0\in\mathfrak{U}_0$ we have
\begin{equation}
\sum_{\Lambda\in\mathcal{F}_{\mathbb{Z}^d}}\left[\Phi(\Lambda),\,A_0\right] = \left[K,\,A_0\right] = 0\,.
\end{equation}
Let us prove that for any $\Lambda\in\mathcal{F}_{\mathbb{Z}^d}$ and any $A_\Lambda\in\mathfrak{U}_\Lambda$ we have
\begin{equation}
\sum_{X\in\mathcal{F}_{\mathbb{Z}^d}}\left[\Phi(X),\,A_\Lambda\right] = 0
\end{equation}
by induction on the size of $\Lambda$.
We have already proved the claim for $|\Lambda|=1$.
Let us fix $x\in\Lambda$.
By linearity, we can assume that $A_\Lambda = A_x\otimes A_{\Lambda\setminus x}$ with $A_{\Lambda\setminus x}\in\mathfrak{U}_{\Lambda\setminus x}$.
We have from the inductive hypothesis
\begin{equation}
\sum_{X\in\mathcal{F}_{\mathbb{Z}^d}}\left[\Phi(X),\,A_x\otimes A_{\Lambda\setminus x}\right] = \sum_{X\in\mathcal{F}_{\mathbb{Z}^d}}\left(\left[\Phi(X),\,A_x\right]A_{\Lambda\setminus x} + A_x\left[\Phi(X),\,A_{\Lambda\setminus x}\right]\right) = 0\,.
\end{equation}
The claim follows.
\end{itemize}
\end{proof}

\section{Duality for the \texorpdfstring{$w_1$}{w\_1} distance}\label{sec:duality}

Using \autoref{thm:alternative-w1}, we prove the following dual formulation for the specific quantum $W_1$ distance.

\begin{thm}[Duality]\label{thm:duality}
For $\rho$, $\sigma \in \mathcal{S}_{\mathbb{Z}^d}^I$, we have
\begin{equation}\label{eq:duality}
w_1(\rho, \sigma) = \sup\left\{\rho(E_{\Phi}) - \sigma(E_{\Phi}): \Phi \in \mathcal{B}_{\mathbb{Z}^d}^r , \, \left\|\Phi\right\|_L\le1\right\}\, .
\end{equation}
\end{thm}

 The result extends \autoref{prop:duality} to infinite spin systems, but unlike the finite dimensional case, in general there is no reason for the supremum in \eqref{eq:duality} to be attained in $\mathcal{B}_{\mathbb{Z}^d}^r$.

\begin{proof}
Let $\Phi \in \mathcal{B}_{\mathbb{Z}^d}^r$ with $\|\Phi\|_L \le 1$, and let $c \ge 0$, $\rho'$, $\sigma' \in \mathcal{S}_{\mathbb{Z}^d}$ be as in \eqref{eq:definition-w1-distance}. For $a \in \mathbb{N}^d_+$ and $\Lambda \subseteq \Lambda_a \in \mathcal{F}_{\mathbb{Z}^d}$, we have
\begin{align}
\rho( \Phi(\Lambda) ) - \sigma( \Phi(\Lambda) ) & = c \sum_{x \in \Lambda_a} \tau_x\rho'( \Phi(\Lambda)) - \tau_x\rho'( \Phi(\Lambda)) \nonumber \\ & = c  \sum_{x \in \Lambda_a}  \rho'( \Phi(\Lambda-x) ) - \sigma'( \Phi(\Lambda-x) ) \, .
\end{align}
If $x \notin \Lambda$, then $0 \notin \Lambda-x$, hence $\rho'( \Phi(\Lambda-x) ) = \sigma'( \Phi(\Lambda-x) )$ because $\rho'_{\Lambda-x} = \sigma'_{\Lambda-x}$.  Therefore, the sum above can be restricted upon $x \in \Lambda$, \emph{i.e.}, $0 \in \Lambda-x$. We then sum upon $\Lambda \in\mathcal{F}_{\mathbb{Z}^d}$ and  make a change of variable $\Lambda' = \Lambda - x$, obtaining
\begin{align}\label{eq:energy-change-variables}
\sum_{0 \in \Lambda \subseteq \Lambda_a } \frac{  \rho(  \Phi(\Lambda)  ) - \sigma(\Phi(\Lambda)) }{|\Lambda|} & = c \sum_{0 \in \Lambda \subseteq \Lambda_a } \sum_{x \in \Lambda} \frac{ \rho'( \Phi(\Lambda-x) ) - \sigma'( \Phi(\Lambda-x) ) }{|\Lambda|} \nonumber \\
& = c \sum_{0 \in \Lambda' \in \mathcal{F}_{\mathbb{Z}^d} } \left( \rho'( \Phi(\Lambda') ) - \sigma'( \Phi(\Lambda') )\right) \frac{ g_a(\Lambda') }{|\Lambda'|} ,
\end{align}
where  $g_a(\Lambda')$ denotes the number of pairs $(x,\Lambda)$ with $x \in \Lambda \subseteq \Lambda_a$, $0 \in \Lambda$, such that $\Lambda' = \Lambda-x$. Notice that the above is a finite sum, since we may restrict upon $\Lambda' \subseteq \Lambda_{2a}$, otherwise $g_a(\Lambda') = 0$. Moreover, for every such pair $(x,\Lambda)$, it must be  $x \in -\Lambda'$, since $0 \in \Lambda$. Therefore, for every $\Lambda' \in \mathcal{F}_{\mathbb{Z}^d}$ with $0 \in \Lambda'$,
\begin{equation}
0 \le \frac{ g_a(\Lambda') }{|\Lambda'|} \le 1\, , \qquad \text{and} \qquad \lim_{a \to \infty} \frac{g_a(\Lambda')}{|\Lambda'|} =1\, ,
\end{equation}
since every pair $(x, \Lambda)$ with $x \in - \Lambda'$ and $\Lambda =x+\Lambda'$ satisfies $x \in \Lambda \subseteq \Lambda_a$ if $a$ is sufficiently large.  Therefore, by the dominated convergence theorem for series, we deduce that
\begin{equation}
\lim_{a \to \infty}  \sum_{0 \in \Lambda' \in \mathcal{F}_{\mathbb{Z}^d} } \rho'( \Phi(\Lambda') ) \frac{ g_a(\Lambda') }{|\Lambda'|} =  \sum_{0 \in \Lambda' \in \mathcal{F}_{\mathbb{Z}^d} } \rho'( \Phi(\Lambda') )\, ,
\end{equation}
and similarly for $\sigma'$. The left hand side in \eqref{eq:energy-change-variables} converges to $\rho(E_\Phi) - \sigma(E_{\Phi})$ as $a \to \infty$, hence we obtain the identity
\begin{equation}\label{eq:duality-identity}
\rho(E_\Phi) - \sigma(E_{\Phi}) = c  \sum_{0 \in \Lambda \in \mathcal{F}_{\mathbb{Z}^d} } \rho'( \Phi(\Lambda) ) -  \sum_{0 \in \Lambda \in \mathcal{F}_{\mathbb{Z}^d} } \sigma'( \Phi(\Lambda) ).
\end{equation}
Given $A \in \mathcal{O}_{\mathbb{Z}^d\setminus 0}$, since $\rho'_{\mathbb{Z}^d\setminus 0} = \sigma'_{\mathbb{Z}^d\setminus 0}$, we have $\rho'( A ) = \sigma'(A)$, thus
\begin{align}
 \sum_{0 \in \Lambda \in \mathcal{F}_{\mathbb{Z}^d} } \rho'( \Phi(\Lambda) ) -  \sum_{0 \in \Lambda \in \mathcal{F}_{\mathbb{Z}^d} } \sigma'( \Phi(\Lambda) ) & =  \rho'\left( \sum_{0 \in \Lambda \in \mathcal{F}_{\mathbb{Z}^d} }  \Phi(\Lambda) - A \right) -  \sigma'\left( \sum_{0 \in \Lambda \in \mathcal{F}_{\mathbb{Z}^d} }  \Phi(\Lambda) - A \right)\\
 & \le 2 \left\|\sum_{0 \in \Lambda \in \mathcal{F}_{\mathbb{Z}^d} }  \Phi(\Lambda) - A  \right\|_\infty\, .
 \end{align}
 Being $A \in \mathcal{O}_{\mathbb{Z}^d\setminus 0}$ arbitrary, we deduce the inequality
\begin{equation}
\sum_{0 \in \Lambda \in \mathcal{F}_{\mathbb{Z}^d} } \rho'( \Phi(\Lambda) ) -  \sum_{0 \in \Lambda \in \mathcal{F}_{\mathbb{Z}^d} } \sigma'( \Phi(\Lambda) ) \le \| \Phi\|_{L} \le 1\, ,
\end{equation}
which from \eqref{eq:duality-identity} gives
\begin{equation}
\sup\left\{ \rho(E_{\Phi}) - \sigma(E_{\Phi}): \Phi \in \mathcal{B}_{\mathbb{Z}^d}^r , \, \left\|\Phi\right\|_L\le1\right\} \le c\, .
\end{equation}
Recalling that $c$ and $\rho'$, $\sigma'$ are chosen as in \eqref{eq:definition-w1-distance}, we deduce
\begin{equation}
\sup\left\{ \rho(E_{\Phi}) - \sigma(E_{\Phi}): \Phi \in \mathcal{B}_{\mathbb{Z}^d}^r , \, \left\|\Phi\right\|_L\le1\right\} \le \| \rho - \sigma\|_{w_1} ,
\end{equation}
\emph{i.e.}, inequality $\ge$ holds in \eqref{eq:duality}.

For the converse inequality, given any $H \in \mathcal{O}_{\Lambda_a}$ with $\|H\|_L \le 1$, we  define the translation invariant interaction
\begin{equation}
\Phi^H(\Lambda) = \frac{ \tau_x H}{|\Lambda_a|}
\end{equation}
 if $\Lambda = \Lambda_a + x$ for some $x \in \mathbb{Z}^d$, $\Phi^H(\Lambda) = 0$ otherwise. Notice that $\Phi^H \in \mathcal{B}_{\mathbb{Z}^d}^r$ and
\begin{equation}
\| \Phi^H \|_L \le \sum_{x \in \Lambda_a} \frac{ \partial_0 \tau_{-x} H}{|\Lambda_a|} \le \sum_{x \in \Lambda_a} \frac{ \partial_{x} H}{|\Lambda_a|} \le \| H\|_L \le 1\, .
\end{equation}
Since $\rho$, $\sigma \in \mathcal{S}_{\mathbb{Z}^d}^I$, we have $\rho(\tau_xH) = \rho(H)$, $\sigma(\tau_x H) = \sigma(H)$, hence
 \begin{equation}
 \rho(E_{\Phi^H}) - \sigma(E_{\Phi^H}) = \sum_{x \in \Lambda_a} \frac{ \rho( \tau_{-x} H ) - \sigma(\tau_{-x} H)}{|\Lambda_a|^2} = \frac{ \rho(H) - \sigma(H)}{|\Lambda_a|}\, .
 \end{equation}
The duality for the quantum $W_1$ distance on the finite lattice $\Lambda_a$ yields
\begin{align}
 \frac{  \| \rho _{\Lambda_a}- \sigma_{\Lambda_a} \|_{W_1} }{|\Lambda_a|} & = \sup\left\{ \rho(E_{\Phi^H}) - \sigma(E_{\Phi^H}): H \in \mathcal{O}_{\Lambda_a} , \, \left\|H\right\|_L\le1\right\} \nonumber \\
  & \le \sup\left\{ \rho(E_{\Phi}) - \sigma(E_{\Phi}): \Phi \in \mathcal{B}_{\mathbb{Z}^d}^r , \, \left\|\Phi\right\|_L\le1\right\}\, .
 \end{align}
 Letting $a \to \infty$, we obtain inequality $\le$ in \eqref{eq:duality}, hence the thesis.
\end{proof}

\section{Recovery of Ornstein's \texorpdfstring{$\bar{d}$}{dbar}-distance}\label{sec:recovery}

As in the finite dimensional case, the specific quantum $W_1$ distance between states recovers Ornstein's $\bar{d}$-distance, when restricted to diagonal states in the canonical basis, \emph{i.e.}, $\rho \in \mathcal{S}_{\mathbb{Z}^d}$ such that,  for every $\Lambda \in \mathcal{F}_{\mathbb{Z}^d}$,  $\rho_{\Lambda}$ is diagonal in the basis $( |x \rangle \langle x |)_{x \in [q]^{\Lambda}}$.

There is indeed a correspondence between probability measures $\mu$ on $[q]^\mathbb{Z^d}$ and such states, defined by mapping $\mu$ to the diagonal state $\rho \in \mathcal{S}_{\mathbb{Z}^d}$ such that, for every $\Lambda \in \mathcal{F}_{\mathbb{Z}^d}$,
\begin{equation}\label{eq:rhop}
\rho_{\Lambda} = \sum_{x \in [q]^{\Lambda}} \mu_{\Lambda}(x)\, | x \rangle \langle x|\, ,
\end{equation}
where $\mu_{\Lambda}$ denotes the marginal of $\mu$ on $\Lambda$. Since states are determined by their collection of marginals, \eqref{eq:rhop} completely determines $\rho$.

The correspondence is clearly invertible, arguing similarly on the space of probability measures  $[q]^\mathbb{Z^d}$. With a slight abuse of notation, we write $|x \rangle \langle x| \in \mathcal{S}_{\mathbb{Z}^d}$ for the diagonal state corresponding to  the Dirac probability measure concentrated at $x \in [q]^{\mathbb{Z}^d}$, so that one can also write
\begin{equation} \rho = \int_{[q]^{\mathbb{Z}^d} } | x \rangle \langle x|\,d \mu(x)\, ,\end{equation}
where integration is in the sense of Pettis (also called weak integral).

Given two shift-invariant (\emph{i.e.}, stationary) probability measures $\mu$, $\nu$ on the infinite product space $[q]^{\mathbb{Z}^d}$, Ornstein's $\bar{d}$-distance \cite{ornstein1973application, gray1975generalization} is defined as
\begin{equation} \label{eq:ornstein} \bar{d} (\mu, \nu) = \sup_{a \in \mathbb{N}^d_+} \frac{W_{1}\left( \mu_{\Lambda_a}, \nu_{\Lambda_a}\right)}{|\Lambda_a|}\, , \end{equation}
where, for $\Lambda \in \mathcal{F}_{\mathbb{Z}^d}$, $W_{1}$ denotes the classical optimal transport distance with Hamming cost on $[q]^{\Lambda}$, \emph{i.e.},
\begin{equation}
W_{1}\left( \mu_{\Lambda}, \nu_{\Lambda} \right) = \min_{\pi \in \mathcal{C}(\mu_{\Lambda}, \nu_{\Lambda})} \sum_{x, y \in [q]^{\Lambda}} h(x,y) \pi(x,y),
\end{equation}
with $\mathcal{C}(\mu_{\Lambda}, \nu_{\Lambda})$ being the set of couplings between the probability distributions $\mu_{\Lambda}$, $\nu_{\Lambda}$, and
\begin{equation}
h(x, y) =  \left| \left\{ i \in \Lambda \, : \, x_i \neq y_i \right\} \right|\, .
\end{equation}

This distance is usually defined only in the case $d=1$, but the extension to $d \ge 1$ is straightforward and informally discussed already in \cite[Appendix 4]{ornstein1973application}.

\begin{prop}\label{prop:recovery}
Given stationary probability measures $\mu$, $\nu$ on $[q]^{\mathbb{Z}^d}$, let $\rho$, $\sigma \in \mathcal{S}_{\mathbb{Z}^d}^I$ denote the associated diagonal states,
\begin{equation} \label{eq:diagonal-states} \rho = \int_{[q]^{\mathbb{Z}^d}} | x \rangle \langle x| d \mu(x), \quad  \sigma = \int_{[q]^{\mathbb{Z}^d}} | x \rangle \langle x| d \nu(x)\, .\end{equation}
Then, we have
\begin{equation} w_1(\rho, \sigma) = \bar{d} (\mu, \nu)\, .\end{equation}
\end{prop}

\begin{proof}
For every $\Lambda \in \mathcal{F}_{\mathbb{Z}^d}$, we have, by \cite[Proposition 5]{de2021quantum},
\begin{equation}
\left\| \rho_{\Lambda} - \sigma_{\Lambda} \right\|_{W_1} = W_1\left( \mu_{\Lambda}, \nu_{\Lambda} \right).
\end{equation}
Choosing $\Lambda = \Lambda_a$, for $a \in \mathbb{N}^d_+$, dividing by $|\Lambda_a|$ and letting $a \to \infty$ yields  the thesis.
\end{proof}

Ornstein's $\bar{d}$-distance \eqref{eq:ornstein} can be equivalently  defined \cite[Theorem 1]{gray1975generalization} as
\begin{equation}
\bar d\left( \mu, \nu \right) = \min_{\pi \in \mathcal{C}^I(\mu, \nu)} \sum_{x, y \in [q]} h(x,y) \pi_0(x,y)\, ,
\end{equation}
where $\mathcal{C}^I(\mu,\nu)$ denotes the set of stationary couplings between the probability distributions $\mu$, $\nu$, and $\pi_0$ is the marginal density of $\pi$ at $0$. \autoref{def:alternative-wass} together with \autoref{thm:alternative-w1} above provide a similar characterization for quantum spin systems, where  stationary couplings are replaced in \eqref{eq:definition-w1-distance}  by representations of the difference the states as series of translates. In fact, if the states $\rho$, $\sigma$ are diagonal, we can also restrict minimization in \eqref{eq:definition-w1-distance} to diagonal states $\rho'$, $\sigma'$, corresponding to probability measures $\mu'$, $\nu'$, obtaining the following further equivalent representation of Ornstein's distance.

\begin{cor}
Given stationary probability measures $\mu$, $\nu$ on  $[q]^{\mathbb{Z}^d}$, we have
\begin{align} \bar{d} (\mu, \nu)=  \min \Bigg\{ c\ge0 : &  \text{\, $\exists \, \mu', \, \nu'$ probability measures on $[q]^{\mathbb{Z}^d}$} :\; \mu'_{\mathbb{Z}^d \setminus 0} = \nu'_{\mathbb{Z}^d\setminus 0}\,,\nonumber\\
&\mu_\Lambda - \nu_\Lambda = c \sum_{x \in \Lambda } (\tau_x \mu')_{\Lambda} - (\tau_x \nu')_{\Lambda}\quad\forall \, \Lambda \in \mathcal{F}_{\mathbb{Z}^d}\Bigg\}\,.
 \end{align}
\end{cor}

To our knowledge, duality for Orstein's $\bar{d}$-distance is not explicitly discussed in the literature. A result can be obtained directly from \autoref{thm:duality} for diagonal states, simply noticing that  the supremum may run among interactions $\Phi$ such that each $\Phi(\Lambda)$ is also diagonal, i.e., corresponding to a function
\begin{equation}
f(\Lambda): [q]^{\mathbb{Z}^d} \to \mathbb{R}\, ,
\end{equation}
depending only on the coordinates in $\Lambda$. Let us denote by $\mathcal{B}^{r, \operatorname{diag}}_{\mathbb{Z}^d}$ the set of translation invariant diagonal interactions satisfying \eqref{eq:defBr}. The Lipschitz constant of $f \in\mathcal{B}^{r, \operatorname{diag}}_{\mathbb{Z}^d}$ coincides with the oscillation of the function on $[q]^{\mathbb{Z}^d}$,
\begin{equation}
x \mapsto \sum_{ 0 \in \Lambda \in \mathcal{F}_{\mathbb{Z}^d}} f(\Lambda)(x)
\end{equation}
with respect to the $0$-th coordinate, i.e.,
\begin{equation}
 \| f \|_{L} = \sup \left\{ \sum_{ 0 \in \Lambda \in \mathcal{F}_{\mathbb{Z}^d}}\left( f(\Lambda) (x) - f(\Lambda) (y) \right) \, : \, \text{ $x$, $y \in [q]^{\mathbb{Z}^d}$,  $x_k=y_k$ for every $k \in \mathbb{Z}^d \setminus 0$}  \right\} \, .
\end{equation}
The specific energy of $f$ is identified with the function on $[q]^{\mathbb{Z}^d}$,
\begin{equation}
 x \mapsto e_f(x) = \sum_{0\in\Lambda\in\mathcal{F}_{\mathbb{Z}^d}}\frac{f(\Lambda)(x)}{\left|\Lambda\right|}\, .
\end{equation}

With this notation, \autoref{thm:duality} yields the following result.

\begin{cor}
Given stationary probability measures $\mu$, $\nu$ on $[q]^\mathbb{Z^d}$, we have
\begin{equation}
\bar{d}(\mu,\nu) = \sup\left\{ \int_{[q]^{\mathbb{Z}^d}} e_f(x) d \mu (x)- \int_{[q]^{\mathbb{Z}^d}} e_f(x) d \nu(x): f \in \mathcal{B}_{\mathbb{Z}^d}^{r,\operatorname{diag}} , \, \left\|f \right\|_L\le1\right\} \, .
\end{equation}
\end{cor}

\section{Quantum concentration inequalities}\label{sec:conc}

\subsection{Poincar\'e inequality}\label{sec:Poinc}
In this section, we prove the following quantum Poincar\'e inequality stating that for any interaction $\Phi$, the variance of the local Hamiltonian on $\Lambda_a$ on a product state scales at most linearly with the volume of $\Lambda_a$ in the limit $a\to\infty$, and the proportionality constant is upper bounded by the square of the Lipschitz constant of $\Phi$:
\begin{prop}[Poincar\'e inequality]\label{prop:PoincI}
Let $\omega\in\mathcal{S}_{\mathbb{Z}^d}$ be a product state.
Then, for any interaction $\Phi\in\mathcal{B}_{\mathbb{Z}^d}^r$ we have
\begin{equation}
\limsup_{a\to\infty}\frac{\mathrm{Var}_{\omega_{\Lambda_a}}H^\Phi_{\Lambda_a}}{\left|\Lambda_a\right|} \le \left\|\Phi\right\|_L^2\,.
\end{equation}
\end{prop}

\subsubsection{Proof of \autoref{prop:PoincI}}
The proof of \autoref{prop:PoincI} is based on its counterpart for quantum spin systems on finite lattices:
\begin{prop}[Quantum Poincar\'e inequality {\cite[Lemma F.1]{de2022limitations}}]\label{prop:Poinc}
Let $\Lambda$ be a finite set, and let $\omega\in\mathcal{S}_\Lambda$ be a product state.
Then, for any $H\in\mathcal{O}_\Lambda$ we have
\begin{equation}
\mathrm{Var}_\omega H = \mathrm{Tr}\left[\omega\left(H - \mathrm{Tr}\left[\omega\,H\right]\mathbb{I}\right)^2\right] \le \sum_{x\in\Lambda}\left(\partial_x H\right)^2\,.
\end{equation}
\end{prop}

The key step in the proof of \autoref{prop:PoincI} is the following proposition, which states that the Lipschitz constant of $\Phi$ is equal to the quadratic average of the dependence of the local Hamiltonians of $\Phi$ on the sites in the limit of infinite volume:
\begin{prop}\label{prop:gradnorm}
Let $\Phi\in\mathcal{B}_{\mathbb{Z}^d}^r$.
Then,
\begin{equation}
\lim_{a\to\infty}\frac{1}{\left|\Lambda_a\right|}\sum_{x\in\Lambda_a}\left(\partial_x H^\Phi_{\Lambda_a}\right)^2 = \left\|\Phi\right\|_L^2\,.
\end{equation}
\end{prop}

\begin{proof}
We fix $\epsilon>0$.
Let $N\in\mathbb{N}$ and $0\in\Lambda_1\in\mathcal{F}_{\mathbb{Z}^d},\,\ldots,\,0\in\Lambda_N\in\mathcal{F}_{\mathbb{Z}^d}$ such that
\begin{equation}
\sum_{0\in\Lambda\in\mathcal{F}_{\mathbb{Z}^d},\,\Lambda\neq\Lambda_1,\,\ldots,\,\Lambda_N}\left\|\Phi(\Lambda)\right\|_\infty < \epsilon\,.
\end{equation}
Let $a\in\mathbb{N}_+^d$ such that $\Lambda_1\cup\ldots\cup\Lambda_N\subseteq\Lambda_a$ and
\begin{equation}
\left|\partial_0 H^\Phi_{\Lambda_{a}} - \left\|\Phi\right\|_L\right| < 2\,\epsilon\,.
\end{equation}
We have
\begin{equation}
\sum_{0\in\Lambda\in\mathcal{F}_{\mathbb{Z}^d},\,\Lambda\not\subseteq\Lambda_a}\left\|\Phi(\Lambda)\right\|_\infty < \epsilon\,.
\end{equation}
Let us fix $x\in\mathbb{Z}^d$.
For any $b\in\mathbb{N}_+^d$ with $b\ge a \pm x$ we have $\Lambda_a\subseteq\Lambda_b-x$ and
\begin{align}\label{eq:Cauchy}
\left|\partial_x H^\Phi_{\Lambda_b} - \partial_0 H^\Phi_{\Lambda_a}\right| &\overset{\mathrm{(a)}}{=} \left|\partial_0 H^\Phi_{\Lambda_b-x} - \partial_0 H^\Phi_{\Lambda_a}\right| \overset{\mathrm{(b)}}{\le} \partial_0\left(H^\Phi_{\Lambda_b-x} - H^\Phi_{\Lambda_a}\right) = \partial_0\sum_{\Lambda\subseteq\Lambda_b-x,\,\Lambda\not\subseteq\Lambda_a}\Phi(\Lambda) \nonumber\\
&= \partial_0\sum_{0\in\Lambda\subseteq\Lambda_b-x,\,\Lambda\not\subseteq\Lambda_a}\Phi(\Lambda)
\le 2\sum_{0\in\Lambda\subseteq\Lambda_b-x,\,\Lambda\not\subseteq\Lambda_a}\left\|\Phi(\Lambda)\right\|_\infty < 2\,\epsilon\,,
\end{align}
where (a) follows from the translation invariance of $\Phi$ and (b) follows since $\partial_0$ is a seminorm.
We then have
\begin{equation}
\left|\partial_x H^\Phi_{\Lambda_{b}} - \left\|\Phi\right\|_L\right| < 4\,\epsilon\,.
\end{equation}
Let $b>a$, such that any $x\in\Lambda_{b-a}$ satisfies $b\ge a\pm x$.
We have from \autoref{lem:partialxH}
\begin{align}
\sum_{x\in\Lambda_b}\left(\partial_x H^\Phi_{\Lambda_b}\right)^2 &= \sum_{x\in\Lambda_{b-a}}\left(\partial_x H^\Phi_{\Lambda_b}\right)^2 + \sum_{x\in\Lambda_b\setminus\Lambda_{b-a}}\left(\partial_x H^\Phi_{\Lambda_b}\right)^2\nonumber\\
&\le \left|\Lambda_{b-a}\right|\left(\left\|\Phi\right\|_L + 4\,\epsilon\right)^2 + 4\left(\left|\Lambda_b\right| - \left|\Lambda_{b-a}\right|\right)\left\|\Phi\right\|_r^2\,,
\end{align}
therefore
\begin{align}
\limsup_{b\to\infty} \frac{1}{\left|\Lambda_b\right|}\sum_{x\in\Lambda_b}\left(\partial_x H^\Phi_{\Lambda_b}\right)^2 &\le \limsup_{b\to\infty}\left(\frac{\left|\Lambda_{b-a}\right|}{\left|\Lambda_b\right|}\left(\left\|\Phi\right\|_L + 4\,\epsilon\right)^2 + 4\left(1 - \frac{\left|\Lambda_{b-a}\right|}{\left|\Lambda_b\right|}\right)\left\|\Phi\right\|_r^2\right)\nonumber\\
&=\left(\left\|\Phi\right\|_L + 4\,\epsilon\right)^2\,.
\end{align}
Since $\epsilon$ is arbitrary, we get
\begin{equation}\label{eq:limsupPoinc}
\limsup_{b\to\infty} \frac{1}{\left|\Lambda_b\right|}\sum_{x\in\Lambda_b}\left(\partial_x H^\Phi_{\Lambda_b}\right)^2 \le \left\|\Phi\right\|_L^2\,.
\end{equation}
If $\left\|\Phi\right\|_L = 0$, the claim follows from \eqref{eq:limsupPoinc}.
If $\left\|\Phi\right\|_L > 0$, we can choose $\epsilon < \left\|\Phi\right\|_L/4$.
We then have
\begin{equation}
\sum_{x\in\Lambda_b}\left(\partial_x H^\Phi_{\Lambda_b}\right)^2 \ge \sum_{x\in\Lambda_{b-a}}\left(\partial_x H^\Phi_{\Lambda_b}\right)^2 \ge \left|\Lambda_{b-a}\right|\left(\left\|\Phi\right\|_L - 4\,\epsilon\right)^2\,,
\end{equation}
and
\begin{equation}
\liminf_{b\to\infty} \frac{1}{\left|\Lambda_b\right|}\sum_{x\in\Lambda_b}\left(\partial_x H^\Phi_{\Lambda_b}\right)^2 \ge \liminf_{b\to\infty}\frac{\left|\Lambda_{b-a}\right|}{\left|\Lambda_b\right|}\left(\left\|\Phi\right\|_L - 4\,\epsilon\right)^2 = \left(\left\|\Phi\right\|_L - 4\,\epsilon\right)^2\,.
\end{equation}
Since $\epsilon$ is arbitrary, we get
\begin{equation}
\liminf_{b\to\infty} \frac{1}{\left|\Lambda_b\right|}\sum_{x\in\Lambda_b}\left(\partial_x H^\Phi_{\Lambda_b}\right)^2 \ge \left\|\Phi\right\|_L^2\,.
\end{equation}
The claim follows.
\end{proof}

We can now conclude the proof of \autoref{prop:PoincI}.
We have from \autoref{prop:Poinc} and \autoref{prop:gradnorm}
\begin{equation}
\limsup_{a\to\infty}\frac{\mathrm{Var}_{\omega_{\Lambda_a}}H^\Phi_{\Lambda_a}}{\left|\Lambda_a\right|} \le \limsup_{a\to\infty}\frac{1}{\left|\Lambda_a\right|}\sum_{x\in\Lambda_a}\left(\partial_x H^\Phi_{\Lambda_a}\right)^2 = \left\|\Phi\right\|_L^2\,.
\end{equation}
The claim follows.

\subsection{Gaussian concentration inequality}\label{sec:Gaussconc}
In this section we prove the following Gaussian concentration inequality (\autoref{thm:Gauss}) for quantum spin systems on finite lattices and apply it to prove an upper bound to the pressure of an interaction in terms of its Lipschitz constant (\autoref{cor:PL}).

\begin{thm}[Gaussian concentration inequality]\label{thm:Gauss}
Let $\Lambda$ be a finite set and let $\omega\in\mathcal{S}_\Lambda$ be a product state with full support.
Then, for any $H\in\mathcal{O}_\Lambda$ we have
\begin{equation}\label{eq:Poincexp}
\ln\mathrm{Tr}_\Lambda e^{H + \ln\omega} \le \mathrm{Tr}_\Lambda\left[\omega\,H\right] + \frac{1}{2}\sum_{x\in\Lambda}\left(\partial_x H\right)^2\,.
\end{equation}
\end{thm}

\begin{rem}
Ref. \cite{de2021quantum} proved the following Gaussian concentration inequality:
\end{rem}

\begin{thm}[{\cite[Theorem 3]{de2021quantum}}]
Let $\Lambda$ be a finite set and let $\omega = \frac{\mathbb{I}_\Lambda}{q^{|\Lambda|}}\in\mathcal{S}_\Lambda$ be the uniform distribution.
Then, for any $H\in\mathcal{O}_\Lambda$ we have
\begin{equation}\label{eq:Poincexpold}
\ln\mathrm{Tr}_\Lambda e^{H + \ln\omega} \le \mathrm{Tr}_\Lambda\left[\omega\,H\right] + \frac{\left|\Lambda\right|}{8}\left\|H\right\|_L^2\,.
\end{equation}
\end{thm}

Upon replacing the constant $\frac{1}{8}$ by $\frac{1}{2}$, the inequality \eqref{eq:Poincexpold} is implied by \eqref{eq:Poincexp}.

\begin{proof}
We will prove the claim by induction on the size of $\Lambda$.
For $\Lambda=\emptyset$ equality holds in \eqref{eq:Poincexp}.
Let us fix $x\in\Lambda$, and let $\Lambda_0 = \Lambda\setminus x$.
Let $H_{\Lambda_0}\in\mathcal{O}_{\Lambda_0}$ such that
\begin{equation}
\partial_x H = 2\left\|H - H_{\Lambda_0}\right\|_\infty\,.
\end{equation}
We have
\begin{equation}
\left\|H - \mathrm{Tr}_x\left[\omega_x\,H\right]\right\|_\infty = \left\|H - H_{\Lambda_0} - \mathrm{Tr}_x\left[\omega_x\left(H - H_{\Lambda_0}\right)\right]\right\|_\infty \le 2\left\|H - H_{\Lambda_0}\right\|_\infty = \partial_x H\,.
\end{equation}
Using the inequality
\begin{equation}
e^t \le \frac{\sinh a}{a}\,t + e^\frac{a^2}{2}\,,\qquad |t|\le a\,,
\end{equation}
we get
\begin{equation}
e^{H - \mathrm{Tr}_x\left[\omega_x H\right]} \le \frac{\sinh\partial_x H}{\partial_x H}\left(H - \mathrm{Tr}_x\left[\omega_x\,H\right]\right) + e^\frac{\left(\partial_x H\right)^2}{2}\,,
\end{equation}
therefore
\begin{equation}\label{eq:Mart}
\mathrm{Tr}_x\left[\omega_x\,e^{H - \mathrm{Tr}_x\left[\omega_x H\right]}\right] \le e^\frac{\left(\partial_x H\right)^2}{2}\,.
\end{equation}
For any $y\in\Lambda_0$, let $H_{\Lambda\setminus y}\in\mathcal{O}_{\Lambda\setminus y}$ such that
\begin{equation}
\partial_y H = 2\left\|H - H_{\Lambda\setminus y}\right\|_\infty\,.
\end{equation}
We have
\begin{equation}\label{eq:partialtr}
\partial_y \mathrm{Tr}_x\left[\omega_x H\right] \le 2\left\|\mathrm{Tr}_x\left[\omega_x H\right] - \mathrm{Tr}_x\left[\omega_x H_{\Lambda\setminus y}\right]\right\|_\infty \le 2\left\|H - H_{\Lambda\setminus y}\right\|_\infty = \partial_y H\,.
\end{equation}
We then have
\begin{align}\label{eq:chain}
\ln\mathrm{Tr}_\Lambda e^{H + \ln \omega} &= \ln\mathrm{Tr}_\Lambda\exp\left(H - \mathrm{Tr}_x\left[\omega_x H\right] + \ln\omega_x + \mathrm{Tr}_x\left[\omega_x H\right] + \ln\omega_{\Lambda_0}\right) \nonumber\\
&\overset{\mathrm{(a)}}{\le} \ln\int_0^\infty\mathrm{Tr}_\Lambda\left[e^{\mathrm{Tr}_x\left[\omega_x H\right] + \ln\omega_{\Lambda_0}} \left(\omega_x^{-1}+t\right)^{-1} e^{H - \mathrm{Tr}_x\left[\omega_x H\right]}\left(\omega_x^{-1}+t\right)^{-1}\right]dt\nonumber\\
&= \ln\mathrm{Tr}_{\Lambda_0}\left[e^{\mathrm{Tr}_x\left[\omega_x H\right] + \ln\omega_{\Lambda_0}}\,\mathrm{Tr}_x\left[e^{H - \mathrm{Tr}_x\left[\omega_x H\right]}\int_0^\infty\left(\omega_x^{-1}+t\right)^{-2}dt\right]\right]\nonumber\\
&= \ln\mathrm{Tr}_{\Lambda_0}\left[e^{\mathrm{Tr}_x\left[\omega_x H\right] + \ln\omega_{\Lambda_0}}\,\mathrm{Tr}_x\left[\omega_x\,e^{H - \mathrm{Tr}_x\left[\omega_x H\right]}\right]\right]\nonumber\\
&\overset{\mathrm{(b)}}{\le} \frac{\left(\partial_x H\right)^2}{2}+ \ln\mathrm{Tr}_{\Lambda_0}e^{\mathrm{Tr}_x\left[\omega_x H\right] + \ln\omega_{\Lambda_0}}\nonumber\\
&\overset{\mathrm{(c)}}{\le} \frac{\left(\partial_x H\right)^2}{2} + \mathrm{Tr}_\Lambda\left[\omega\,H\right] + \frac{1}{2}\sum_{y\in\Lambda_0}\left(\partial_y \mathrm{Tr}_x\left[\omega_x H\right]\right)^2 \overset{\mathrm{(d)}}{\le} \mathrm{Tr}_\Lambda\left[\omega\,H\right] + \frac{1}{2}\sum_{y\in\Lambda}\left(\partial_y H\right)^2\,,
\end{align}
where (a) follows from the Golden--Thompson inequality with three matrices \cite{lieb1973convex}, (b) from \eqref{eq:Mart}, (c) from the inductive hypothesis and (d) from \eqref{eq:partialtr}.
The claim follows.
\end{proof}

\begin{cor}\label{cor:PL}
Let $\Phi\in\mathcal{B}^r_{\mathbb{Z}^d}$ and let $\omega\in\mathcal{S}^I_{\mathbb{Z}^d}$ be the uniform distribution, \emph{i.e.}, $\omega_\Lambda = \frac{\mathbb{I}_\Lambda}{q^{|\Lambda|}}$ for any $\Lambda\in\mathcal{F}_{\mathbb{Z}^d}$.
Then,
\begin{equation}
P(\Phi) \le \ln q + \frac{\left\|\Phi\right\|_L^2}{2} - \omega(E_\Phi)\,.
\end{equation}
\end{cor}

\begin{proof}
We have
\begin{align}
P(\Phi) &= \lim_{a\to\infty}\frac{\ln\mathrm{Tr}_{\Lambda_a} e^{-H^\Phi_{\Lambda_a}}}{\left|\Lambda_a\right|} \overset{\mathrm{(a)}}{\le} \ln q + \lim_{a\to\infty}\left(\frac{1}{2\left|\Lambda_a\right|}\sum_{x\in\Lambda_a}\left(\partial_x H^\Phi_{\Lambda_a}\right)^2 - \frac{\omega\left(H^\Phi_{\Lambda_a}\right)}{\left|\Lambda_a\right|}\right)\nonumber\\
&\overset{\mathrm{(b)}}{=} \ln q + \frac{\left\|\Phi\right\|_L^2}{2} - \omega(E_\Phi)\,,
\end{align}
where (a) follows from \autoref{thm:Gauss} and (b) from \autoref{prop:gradnorm} and \eqref{eq:EPhi}.
The claim follows.
\end{proof}

\section{\texorpdfstring{$W_1$}{W\_1} continuity of the von Neumann entropy}\label{sec:continuity}
In this section we prove the following continuity bound of the von Neumann entropy with respect to the quantum $W_1$ distance:

\begin{thm}[$W_1$ continuity of the von Neumann entropy]\label{thm:main}
Let $\Lambda$ be a finite set.
For any $\rho,\,\sigma\in\mathcal{S}_\Lambda$ we have
\begin{equation}\label{eq:main}
\frac{\left|S(\rho) - S(\sigma)\right|}{\left|\Lambda\right|} \le h_2\left(\frac{\left\|\rho-\sigma\right\|_{W_1}}{\left|\Lambda\right|}\right) + \frac{\left\|\rho - \sigma\right\|_{W_1}}{\left|\Lambda\right|}\ln\left(q^2-1\right)\,.
\end{equation}
\end{thm}

\autoref{thm:main} generalizes to the quantum setting the following continuity bound of the Shannon entropy with respect to the classical $W_1$ distance:
\begin{thm}[$W_1$ continuity of the Shannon entropy {\cite[Proposition 8]{polyanskiy2016wasserstein}}]
Let $\Lambda$ be a finite set.
For any two probability distributions $\mu,\,\nu$ on $[q]^\Lambda$ we have
\begin{equation}\label{eq:SWc}
\left|S(\mu) - S(\nu)\right|\le \left|\Lambda\right|h_2\left(\frac{W_1(\mu,\nu)}{\left|\Lambda\right|}\right) + W_1(\mu,\nu)\ln\left(q-1\right)\,.
\end{equation}
\end{thm}
The quantum continuity bound \eqref{eq:main} is identical to the classical bound \eqref{eq:SWc} upon replacing $q$ by $q^2$.
Such replacement is necessary, since the von Neumann entropy does not always satisfy the classical bound \eqref{eq:SWc} \cite{de2021quantum}.

\begin{rem}
Ref. \cite{de2021quantum} proved the following weaker continuity bound for the von Neumann entropy in terms of the $W_1$ distance:
\end{rem}

\begin{thm}[{\cite[Theorem 1]{de2021quantum}}]
Let $\Lambda$ be a finite set.
For any $\rho,\,\sigma\in\mathcal{S}_\Lambda$,
\begin{equation}\label{eq:oldbound}
\left|S(\rho) - S(\sigma)\right| \le g\left(\left\|\rho-\sigma\right\|_{W_1}\right) + \left\|\rho - \sigma\right\|_{W_1}\ln\left(q^2\left|\Lambda\right|\right)\,,
\end{equation}
where for any $t\ge0$
\begin{equation}
g(t) = \left(t+1\right)\ln\left(t+1\right) - t\ln t\,.
\end{equation}
\end{thm}
Due to the term $\ln\left|\Lambda\right|$, the bound \eqref{eq:oldbound} does not have the right scaling with respect to $\left|\Lambda\right|$ to prove a continuity bound for the specific entropy in terms of the specific quantum $W_1$ distance.
On the contrary, \autoref{thm:main} will be crucial in the proof of such a bound, which will be the subject of \autoref{sec:w1cont}.

\subsection{Proof of \autoref{thm:main}}

The proof of \autoref{thm:main} is based on the following notion of distance operator:
\begin{defn}[Distance operator {\cite[Section 2]{osborne2009quantum}}, {\cite[Definition 15]{eldar2017local}}]
Let $\mathcal{V}$ be a subspace of $\mathcal{H}_\Lambda$.
For any $k=0,\,\ldots,\,\left|\Lambda\right|$, we define the \emph{fattening} $\mathcal{V}_k$ of $\mathcal{V}$ of radius $k$ as the span of the linear operators acting on at most $k$ sites applied to a vector in $\mathcal{V}$:
\begin{equation}
    \mathcal{V}_k = \mathrm{span}\left\{O|\psi\rangle:|\psi\rangle\in\mathcal{V}\,,\;O\in\mathfrak{U}_X:X\subseteq\Lambda\,,\;|X|\le k\right\}\,,
\end{equation}
such that
\begin{equation}
    \mathcal{V} = \mathcal{V}_0 \subseteq \ldots \subseteq \mathcal{V}_{\left|\Lambda\right|} = \mathcal{H}_\Lambda\,.
\end{equation}
We define the \emph{distance operator} of $\mathcal{V}$ as the linear operator $H_\mathcal{V}\in\mathcal{O}_\Lambda$ that has eigenvalue $k$ on $\mathcal{V}_k\cap\mathcal{V}_{k-1}^\perp$ for each $k=0,\,\ldots,\,\left|\Lambda\right|$.
\end{defn}

The following \autoref{prop:distLip} provides the link between the distance operator and the $W_1$ distance:
\begin{prop}\label{prop:distLip}
Let $\mathcal{V}$ be a subspace of $\mathcal{H}_\Lambda$, and let $\rho,\,\sigma\in\mathcal{S}_\Lambda$ such that the support of $\sigma$ is contained in $\mathcal{V}$.
Then,
\begin{equation}
    \left\|\rho - \sigma\right\|_{W_1} \ge \mathrm{Tr}_\Lambda\left[\rho\,H_\mathcal{V}\right]\,.
\end{equation}
\end{prop}
\begin{proof}
Since $\mathrm{Tr}_\Lambda\left[\sigma\,H_\mathcal{V}\right]=0$, it is sufficient to prove that $\left\|H_\mathcal{V}\right\|_L \le 1$.
For any $k=0,\,\ldots,\,\left|\Lambda\right|$, let $\Pi_k$ be the orthogonal projector onto $\mathcal{V}_k$, such that
\begin{equation}
    H_\mathcal{V} = \sum_{k=0}^{\left|\Lambda\right|}\left(\mathbb{I} - \Pi_k\right)\,.
\end{equation}
For any $x\in\Lambda$, let
\begin{equation}
\mathcal{V}_{k,x} = \mathrm{span}\left\{O|\psi\rangle:|\psi\rangle\in\mathcal{V}\,,\;O\in\mathfrak{U}_X:X\subseteq\Lambda\,,\;|X|\le k\,,\;x\in X\right\}\,,
\end{equation}
and let $\Pi_{k,x}$ be the orthogonal projector onto $\mathcal{V}_{k,x}$.
We have $\mathcal{V}_{k-1} \subseteq \mathcal{V}_{k,x} \subseteq \mathcal{V}_k$, therefore
\begin{equation}
0 \le \Pi_k - \Pi_{k,x} \le \Pi_k - \Pi_{k-1}\,.
\end{equation}
The subspace $\mathcal{V}_{k,x}$ is invariant with respect to the action of any unitary operator $U\in\mathfrak{U}_x$.
Then, $\Pi_{k,x}$ commutes with any such $U$, and therefore $\Pi_{k,x}\in\mathcal{O}_{\Lambda\setminus x}$.
Then,
\begin{align}
\partial_xH_\mathcal{V} &= \partial_x\sum_{k=0}^{\left|\Lambda\right|} \left(\mathbb{I} - \Pi_k\right) = \partial_x\sum_{k=0}^{\left|\Lambda\right|} \Pi_k = \partial_x\sum_{k=0}^{\left|\Lambda\right|}\left(\Pi_k - \Pi_{k,x}\right) \overset{(\mathrm{a})}{\le} \left\|\sum_{k=0}^{\left|\Lambda\right|} \left(\Pi_k - \Pi_{k,x}\right)\right\|_\infty\nonumber\\
&\le \left\|\sum_{k=0}^{\left|\Lambda\right|} \left(\Pi_k - \Pi_{k-1}\right)\right\|_\infty \overset{(\mathrm{b})}{=} 1\,,
\end{align}
where (a) follows from \autoref{lem:PSD} and (b) follows observing that $\Pi_k - \Pi_{k-1}$ is the orthogonal projector onto $\mathcal{V}_k\cap\mathcal{V}_{k-1}^\perp$.
The claim follows.
\end{proof}

We first prove \autoref{thm:main} when $\sigma$ is proportional to an orthogonal projector:

\begin{prop}\label{prop:main}
Let $\mathcal{V}$ be a subspace of $\mathcal{H}_\Lambda$, let $\Pi$ be the associated orthogonal projector and let
\begin{equation}
\sigma = \frac{\Pi}{\dim\mathcal{V}}\,.
\end{equation}
Then, for any $\rho\in\mathcal{S}_\Lambda$ we have
\begin{equation}
S(\rho) - \ln\dim\mathcal{V} \le \left|\Lambda\right|h_2\left(\frac{\left\|\rho-\sigma\right\|_{W_1}}{\left|\Lambda\right|}\right) + \left\|\rho - \sigma\right\|_{W_1}\ln\left(q^2-1\right)\,.
\end{equation}
\end{prop}

\begin{proof}
For any $0\le t \le 1$, let
\begin{equation}
\phi(t) = h_2(t) + t\ln\left(q^2-1\right)\,,
\end{equation}
and let
\begin{equation}
\left\|\rho - \sigma\right\|_{W_1} = \left|\Lambda\right|w\,.
\end{equation}
The claim becomes
\begin{equation}
S(\rho) \le \ln\dim\mathcal{V} + \left|\Lambda\right|\phi(w)\,.
\end{equation}
$\phi$ is increasing in $\left[0,1-\frac{1}{q^2}\right]$ and decreasing in $\left[1-\frac{1}{q^2},1\right]$ with
\begin{equation}
\phi(0) = 0\,,\qquad \phi\left(1-\tfrac{1}{q^2}\right) = \ln q^2\,,\qquad \phi(1) = \ln\left(q^2-1\right)\,.
\end{equation}
Let $0<w^*<1-\frac{1}{q^2}$ satisfy
\begin{equation}\label{eq:w*}
\phi(w^*) = \ln q\,.
\end{equation}
If $w\ge w^*$, the claim is trivial.
Indeed, if $w^*\le w \le 1-\frac{1}{q^2}$ we have
\begin{equation}
\phi(w) \ge \phi(w^*) = \ln q\,,
\end{equation}
while if $1-\frac{1}{q^2}<w\le1$ we have
\begin{equation}
\phi(w) \ge \phi(1) = \ln\left(q^2-1\right) \ge \ln q\,.
\end{equation}
We can then assume $w<w^*$.

Let $H_{\mathcal{V}}$ be the distance operator of $\mathcal{V}$.
For any $k=0,\,\ldots,\,\left|\Lambda\right|$, let $\mathcal{W}_k$ be the eigenspace of $H_{\mathcal{V}}$ with eigenvalue $k$, let $P_k$ be the orthogonal projector onto $\mathcal{W}_k$, and let $p_k = \mathrm{Tr}_\Lambda\left[\rho\,P_k\right]$ be the probability that a measurement of $H_\mathcal{V}$ on $\rho$ has outcome $k$.
Let
\begin{equation}
\tilde{\rho} = \sum_{k=0}^{\left|\Lambda\right|}p_k\,\frac{P_k}{\dim\mathcal{W}_k}\,.
\end{equation}
We have
\begin{equation}\label{eq:Stilde}
0 \le S(\rho\|\tilde{\rho}) = S(\tilde{\rho}) - S(\rho)\,.
\end{equation}
For any $x\in\left\{0,\ldots,q^2-1\right\}^{\Lambda}$, let $H(x)$ be the number of components of $x$ that are different from $0$, and for any $k=0,\,\ldots,\,\left|\Lambda\right|$, let
\begin{equation}
D_k = \left|H^{-1}(k)\right| = \left|\left\{x\in\left\{0,\ldots,q^2-1\right\}^{\Lambda}:H(x)=k\right\}\right|\,.
\end{equation}
Let $X$ be a random variable with values in $\left\{0,\ldots,q^2-1\right\}^{\Lambda}$ distributed as follows.
Let the probability distribution of $H(X)$ be $p$, and for any $k=0,\,\ldots,\,\left|\Lambda\right|$, let the probability distribution of $X$ conditioned on $H(X)=k$ be uniform, such that the probability of $x\in\left\{0,\ldots,q^2-1\right\}^{\Lambda}$ is
\begin{equation}
\mathbb{P}(X=x) = \frac{p_{H(x)}}{D_{H(x)}}\,.
\end{equation}
Since $H(X)$ has the same probability distribution as $H_\mathcal{V}$ measured on $\rho$, we have
\begin{equation}
\mathbb{E}\,H(X) = \mathrm{Tr}_\Lambda\left[\rho\,H_{\mathcal{V}}\right] =: \left|\Lambda\right|u\,.
\end{equation}
By the maximum entropy principle, the Shannon entropy of $X$ is upper bounded by the Shannon entropy of the Gibbs distribution of $H$ with average energy $\left|\Lambda\right|u$ :
\begin{equation}\label{eq:SX}
S(X) \le \left|\Lambda\right|\phi(u)\,.
\end{equation}
We then have
\begin{align}\label{eq:Somega}
S(\rho) &\overset{\mathrm{(a)}}{\le} S(\tilde{\rho}) = \sum_{k=0}^{\left|\Lambda\right|} p_k\ln\frac{\dim\mathcal{W}_k}{p_k} \overset{\mathrm{(b)}}{\le} \sum_{k=0}^{\left|\Lambda\right|} p_k\ln\frac{D_k\dim\mathcal{V}}{p_k} = \ln\dim\mathcal{V} + S(X)\nonumber\\ &\overset{\mathrm{(c)}}{\le} \ln\dim\mathcal{V} + \left|\Lambda\right|\phi(u)\,,
\end{align}
where (a) follows from \eqref{eq:Stilde}, (b) from \autoref{lem:Wk} and (c) from \eqref{eq:SX}.
We have from \autoref{prop:distLip}
\begin{equation}
w \ge \frac{\mathrm{Tr}_\Lambda\left[\rho\,H_{\mathcal{V}}\right]}{\left|\Lambda\right|} = u\,,
\end{equation}
hence
\begin{equation}\label{eq:fu}
\phi(u) \le \phi(w)\,.
\end{equation}
The claim follows.
\end{proof}

Without loss of generality, we can assume $S(\rho)\ge S(\sigma)$.
For any $k\in\mathbb{N}$ and any $\delta>0$, let $P_{k,\delta}$ be the $\delta$-typical projector of $\sigma^{\otimes k}$, \emph{i.e.}, the orthogonal projector on the sum of the eigenspaces of $\sigma^{\otimes k}$ with eigenvalues contained in $\left[e^{-k\left(S(\sigma)+\delta\right)},e^{-k\left(S(\sigma)-\delta\right)}\right]$.
$P_{k,\delta}$ satisfies \cite[Section 5.5]{nielsen2010quantum,wilde2017quantum,holevo2019quantum}
\begin{subequations}
\begin{equation}\label{eq:P1}
\sigma^{\otimes k} \ge e^{-k\left(S(\sigma)+\delta\right)}\,P_{k,\delta}\,,
\end{equation}
\begin{equation}\label{eq:P2}
\liminf_{k\to\infty}\frac{\ln\mathrm{Tr}_\Lambda P_{k,\delta}}{k}\ge S(\sigma) - \delta\,.
\end{equation}
\end{subequations}
The property \eqref{eq:P1} implies
\begin{subequations}
\begin{equation}\label{eq:P3}
\frac{\ln\mathrm{Tr}_\Lambda P_{k,\delta}}{k} \le S(\sigma) + \delta\,,
\end{equation}
\begin{equation}\label{eq:P4}
\frac{1}{k}\,S\left(\left.\frac{P_{k,\delta}}{\mathrm{Tr}_\Lambda P_{k,\delta}}\right\|\sigma^{\otimes k}\right) \le S(\sigma) + \delta - \frac{\ln\mathrm{Tr}_\Lambda P_{k,\delta}}{k}\,.
\end{equation}
\end{subequations}
We have from \eqref{eq:P3} and \autoref{prop:main}
\begin{equation}\label{eq:deltaS}
S(\rho) - S(\sigma) \le \frac{S\left(\rho^{\otimes k}\right) - \ln\mathrm{Tr}_\Lambda P_{k,\delta}}{k} + \delta \le \left|\Lambda\right|\phi\left(\frac{\left\|\rho^{\otimes k} - \frac{P_{k,\delta}}{\mathrm{Tr}_\Lambda P_{k,\delta}}\right\|_{W_1}}{k\left|\Lambda\right|}\right) + \delta\,.
\end{equation}
We have
\begin{align}
\frac{\left\|\rho^{\otimes k} - \frac{P_{k,\delta}}{\mathrm{Tr}_\Lambda P_{k,\delta}}\right\|_{W_1}}{k\left|\Lambda\right|} & \le \frac{\left\|\rho^{\otimes k} - \sigma^{\otimes k}\right\|_{W_1}}{k\left|\Lambda\right|} + \frac{\left\|\sigma^{\otimes k} - \frac{P_{k,\delta}}{\mathrm{Tr}_\Lambda P_{k,\delta}}\right\|_{W_1}}{k\left|\Lambda\right|} \overset{\mathrm{(a)}}{\le} w + \sqrt{\frac{2}{k}\,S\left(\left.\frac{P_{k,\delta}}{\mathrm{Tr}_\Lambda P_{k,\delta}}\right\|\sigma^{\otimes k}\right)}\nonumber\\
&\overset{\mathrm{(b)}}{\le} w + \sqrt{2\left(S(\sigma) + \delta - \frac{\ln\mathrm{Tr}_\Lambda P_{k,\delta}}{k}\right)}\,,
\end{align}
where (a) follows from \autoref{prop:W1SA} and \autoref{prop:Marton}, and (b) follows from \eqref{eq:P4}.
We get from \eqref{eq:P2}
\begin{equation}
\limsup_{k\to\infty}\frac{\left\|\rho^{\otimes k} - \frac{P_{k,\delta}}{\mathrm{Tr}_\Lambda P_{k,\delta}}\right\|_{W_1}}{k\left|\Lambda\right|} \le w + 2\sqrt{\delta}\,.
\end{equation}
We then get from \eqref{eq:deltaS}
\begin{equation}
S(\rho) - S(\sigma) \le \left|\Lambda\right|\phi\left(w + 2\sqrt{\delta}\right) + \delta\,,
\end{equation}
and the claim follows taking the limit $\delta\to0$.

\section{\texorpdfstring{$w_1$}{w\_1} continuity of the specific entropy}\label{sec:w1cont}
A fundamental consequence of \autoref{thm:main} is the following continuity bound for the specific entropy in terms of the specific quantum $W_1$ distance:
\begin{cor}[$w_1$ continuity of the specific entropy]\label{thm:mainI}
The specific entropy satisfies the following continuity bound with respect to the specific quantum $W_1$ distance:
For any $\rho,\,\sigma\in\mathcal{S}_{\mathbb{Z}^d}^I$ we have
\begin{equation}
\left|s(\rho) - s(\sigma)\right| \le h_2\left(w_1(\rho,\sigma)\right) + w_1(\rho,\sigma)\ln\left(q^2-1\right)\,,
\end{equation}
where
\begin{equation}
h_2(t) = -t\ln t - \left(1-t\right)\ln\left(1-t\right)\,,\qquad 0\le t\le1
\end{equation}
is the binary entropy function.
\end{cor}

\begin{proof}
The claim follows from \autoref{thm:main}: We have
\begin{align}
\left|s(\rho) - s(\sigma)\right| &= \lim_{a\to\infty}\frac{\left|S\left(\rho_{\Lambda_a}\right) - S\left(\sigma_{\Lambda_a}\right)\right|}{\left|\Lambda_a\right|}\nonumber\\
&\le \lim_{a\to\infty}\left(h_2\left(\frac{\left\|\rho_{\Lambda_a}-\sigma_{\Lambda_a}\right\|_{W_1}}{\left|\Lambda_a\right|}\right) + \frac{\left\|\rho_{\Lambda_a} - \sigma_{\Lambda_a}\right\|_{W_1}}{\left|\Lambda_a\right|}\ln\left(q^2-1\right)\right)\nonumber\\
&= h_2\left(w_1(\rho,\sigma)\right) + w_1(\rho,\sigma)\ln\left(q^2-1\right)\,.
\end{align}
\end{proof}

\section{\texorpdfstring{$w_1$}{w\_1}-Gibbs states}\label{sec:Gibbs}

We define the $w_1$-Gibbs states of the interaction $\Phi$ as the translation-invariant states whose marginal states have a $W_1$ distance from the local Gibbs states of $\Phi$ that scales sublinearly with the volume:
\begin{defn}[$w_1$-Gibbs state]\label{defn:w1G}
Let $\Phi\in\mathcal{B}^r_{\mathbb{Z}^d}$.
We define for any $\rho\in\mathcal{S}^I_{\mathbb{Z}^d}$ the specific quantum  $W_1$ distance between $\rho$ and $\Phi$ as the limit of the $W_1$ distance per site between the marginals of $\rho$ and the local Gibbs states of $\Phi$:
\begin{equation}
w_1(\rho,\Phi) = \limsup_{a\to\infty}\frac{\left\|\rho_{\Lambda_a} - \omega^\Phi_{\Lambda_a}\right\|_{W_1}}{\left|\Lambda_a\right|}\,.
\end{equation}
We say that the state $\omega\in\mathcal{S}^I_{\mathbb{Z}^d}$ is a \emph{$w_1$-Gibbs state} of $\Phi$ if $w_1(\omega,\Phi) = 0$.
\end{defn}

We also define for any $\Phi\in\mathcal{B}^r_{\mathbb{Z}^d}$ and any $\rho\in\mathcal{S}^I_{\mathbb{Z}^d}$ the specific relative entropy between $\rho$ and $\Phi$ as the limit of the relative entropy per site between the marginals of $\rho$ and the local Gibbs states of $\Phi$:
\begin{equation}
s(\rho\|\Phi) = \lim_{a\to\infty}\frac{S\left(\rho_{\Lambda_a}\left\|\omega^\Phi_{\Lambda_a}\right.\right)}{\left|\Lambda_a\right|} = P(\Phi) - s(\rho) + \rho(E_\Phi)\,.
\end{equation}
We have $s(\rho\|\Phi)\ge0$, with equality iff $\rho\in\mathcal{S}_{eq}(\Phi)$.

An interaction can have at most one $w_1$-Gibbs state:
\begin{prop}[Uniqueness of the $w_1$-Gibbs state]\label{prop:w1G}
Let $\Phi\in\mathcal{B}^r_{\mathbb{Z}^d}$ have a $w_1$-Gibbs state $\omega\in\mathcal{S}^I_{\mathbb{Z}^d}$.
Then, for any $\rho\in\mathcal{S}^I_{\mathbb{Z}^d}$ we have
\begin{equation}
w_1(\rho,\Phi) = w_1(\rho,\omega)\,.
\end{equation}
In particular, $\Phi$ can have at most one $w_1$-Gibbs state.
\end{prop}

\begin{proof}
We have
\begin{align}
\left|w_1(\rho,\Phi) - w_1(\rho,\omega)\right| &= \left|\limsup_{a\to\infty}\frac{\left\|\rho_{\Lambda_a} - \omega^\Phi_{\Lambda_a}\right\|_{W_1} - \left\|\rho_{\Lambda_a} - \omega_{\Lambda_a}\right\|_{W_1}}{\left|\Lambda_a\right|}\right|\nonumber\\
&\le \limsup_{a\to\infty}\frac{\left\|\omega_{\Lambda_a} - \omega^\Phi_{\Lambda_a}\right\|_{W_1}}{\left|\Lambda_a\right|} = w_1(\omega,\Phi) = 0\,.
\end{align}
If also $\rho$ is a $w_1$-Gibbs state of $\Phi$, we have
\begin{equation}
w_1(\rho,\omega) = w_1(\rho,\Phi) = 0\,,
\end{equation}
hence $\rho = \omega$.
The claim follows.
\end{proof}

If an interaction admits a $w_1$-Gibbs state, then such state is also an equilibrium state:
\begin{prop}\label{prop:w1eq}
We have for any $\Phi\in\mathcal{B}^r_{\mathbb{Z}^d}$ and any $\rho\in\mathcal{S}^I_{\mathbb{Z}^d}$
\begin{equation}
s(\rho\|\Phi) \le h_2(w_1(\rho,\Phi)) + w_1(\rho,\Phi)\left(\ln\left(q^2-1\right) + 2\left\|\Phi\right\|_r\right)\,.
\end{equation}
In particular, if $\Phi$ has a $w_1$-Gibbs state $\omega\in\mathcal{S}^I_{\mathbb{Z}^d}$, then $\omega\in\mathcal{S}_{eq}(\Phi)$.
\end{prop}

\begin{proof}
We have
\begin{align}
s(\rho\|\Phi) &= \lim_{a\to\infty}\frac{S\left(\rho_{\Lambda_a}\left\|\omega^\Phi_{\Lambda_a}\right.\right)}{\left|\Lambda_a\right|} = \lim_{a\to\infty}\frac{S\left(\omega^\Phi_{\Lambda_a}\right) - S(\rho_{\Lambda_a})+\mathrm{Tr}_{\Lambda_a}\left[\left(\rho_{\Lambda_a} - \omega^\Phi_{\Lambda_a}\right)H^\Phi_{\Lambda_a}\right]}{\left|\Lambda_a\right|}\nonumber\\
&\overset{\mathrm{(a)}}{\le} \lim_{a\to\infty}\left(h_2\left(\frac{\left\|\rho_{\Lambda_a} - \omega^\Phi_{\Lambda_a}\right\|_{W_1}}{\left|\Lambda_a\right|}\right) + \frac{\left\|\rho_{\Lambda_a} - \omega^\Phi_{\Lambda_a}\right\|_{W_1}}{\left|\Lambda_a\right|}\left(\ln\left(q^2-1\right) + 2\left\|\Phi\right\|_r\right)\right)\nonumber\\
&= h_2(w_1(\rho,\Phi)) + w_1(\rho,\Phi)\left(\ln\left(q^2-1\right) + 2\left\|\Phi\right\|_r\right)\,,
\end{align}
where (a) follows from \autoref{thm:main} and \eqref{eq:LL}\,.
The claim follows.
\end{proof}

\section{Quantum transportation-cost inequalities}\label{sec:TCI}

\begin{defn}[TCI]
The interaction $\Phi\in\mathcal{B}^r_{\mathbb{Z}^d}$ satisfies a \emph{Transportation-Cost Inequality} (TCI) with constant $c>0$ if the square of the specific quantum $W_1$ distance with respect to $\Phi$ is upper bounded by $\frac{c}{2}$ times the specific relative entropy with respect to $\Phi$, \emph{i.e.}, if for any $\rho\in\mathcal{S}^I_{\mathbb{Z}^d}$ we have
\begin{equation}\label{eq:TCI}\tag{TCI}
w_1(\rho,\Phi)^2 \le \frac{c}{2} \, s(\rho\|\Phi)\,.
\end{equation}
\end{defn}

A fundamental consequence of \eqref{eq:TCI} is the uniqueness of the equilibrium state of $\Phi$:

\begin{prop}[Uniqueness of the equilibrium state]\label{thm:uniqueness}
Let $\Phi\in\mathcal{B}^r_{\mathbb{Z}^d}$ satisfy \eqref{eq:TCI}.
Then, $\Phi$ has a unique equilibrium state, which is a $w_1$-Gibbs state.
\end{prop}

\begin{proof}
Let $\omega\in\mathcal{S}_{eq}(\Phi)$.
From \eqref{eq:TCI} we have
\begin{equation}
w_1(\omega,\Phi)^2 \le \frac{c}{2}\,s(\omega\|\Phi) = 0\,,
\end{equation}
therefore $\omega$ is a $w_1$-Gibbs state of $\Phi$.
Since the $w_1$-Gibbs state is unique, the equilibrium state is unique, too.
\end{proof}

Another property of the interactions satisfying \eqref{eq:TCI} is the following upper bound to the variation of the specific entropy in terms of the specific relative entropy:
\begin{prop}\label{prop:ss}
Let $\Phi\in\mathcal{B}^r_{\mathbb{Z}^d}$ satisfy \eqref{eq:TCI} and let $\omega\in\mathcal{S}^I_{\mathbb{Z}^d}$ be its unique equilibrium state.
Let $w^*$ be as in \eqref{eq:w*}.
Then, for any $\rho\in\mathcal{S}_{\mathbb{Z}^d}^I$ such that
\begin{equation}
s(\rho\|\Phi) \le \frac{2\,{w^*}^2}{c}
\end{equation}
we have
\begin{equation}
\left|s(\rho) - s(\omega)\right| \le h_2\left(\sqrt{\frac{c}{2}\,s(\rho\|\Phi)}\right) + \sqrt{\frac{c}{2}\,s(\rho\|\Phi)}\ln\left(q^2-1\right)\,.
\end{equation}
\end{prop}

\begin{proof}
From \autoref{thm:uniqueness}, we have that $\omega$ is a $w_1$-Gibbs state of $\Phi$.
We then have
\begin{align}
\left|s(\rho) - s(\omega)\right| &\overset{\mathrm{(a)}}{\le} h_2(w_1(\rho,\omega)) + w_1(\rho,\omega)\ln\left(q^2-1\right) \nonumber\\
&\overset{\mathrm{(b)}}{\le} h_2\left(\sqrt{\frac{c\,s(\rho\|\Phi)}{2}}\right) + \sqrt{\frac{c\,s(\rho\|\Phi)}{2}}\ln\left(q^2-1\right)\,,
\end{align}
where (a) follows from \autoref{thm:mainI} and (b) from \autoref{prop:w1G} and \eqref{eq:TCI}.
The claim follows.
\end{proof}

In the following, we will prove that \eqref{eq:TCI} is satisfied by interactions containing only single-site terms (\autoref{sec:prod}) and local commuting interactions at high temperature (\autoref{sec:HT}).

\subsection{Product states}\label{sec:prod}

The simplest setting where \eqref{eq:TCI} holds is when the interaction contains only terms acting on single spins and the associated Gibbs state is a product state.
Ref. \cite{de2021quantum} proved the following TCI for product states on finite lattices:
\begin{thm}[Quantum Marton's Transportation Inequality {\cite[Theorem 2]{de2021quantum}}]\label{thm:Marton}
Let $\Lambda$ be a finite set and let $\sigma\in\mathcal{S}_\Lambda$ be a product state.
Then, for any $\rho\in\mathcal{S}_\Lambda$ we have
\begin{equation}
\left\|\rho - \sigma\right\|_{W_1}^2 \le \frac{\left|\Lambda\right|}{2}\,S(\rho\|\sigma)\,.
\end{equation}
\end{thm}

\autoref{thm:Marton} implies the following TCI for product states on $\mathbb{Z}^d$:

\begin{cor}[TCI for product states]\label{cor:prodTCI}
Let $\omega\in\mathcal{S}_{\mathbb{Z}^d}^I$ be a product state.
Then, for any $\rho\in\mathcal{S}_{\mathbb{Z}^d}^I$ we have
\begin{equation}\label{eq:Marton}
w_1(\rho,\omega)^2 \le \frac{1}{2}\,s(\rho\|\omega)\,.
\end{equation}
Therefore, any $\Phi\in\mathcal{B}^r_{\mathbb{Z}^d}$ that contains only terms acting on single spins (\emph{i.e.}, such that $\Phi(\Lambda)=0$ for any $\Lambda\in\mathcal{F}_{\mathbb{Z}^d}$ with $|\Lambda|\ge2$) satisfies \eqref{eq:TCI} with $c=1$.
\end{cor}

\begin{proof}
The claim \eqref{eq:Marton} follows from \autoref{thm:Marton}: We have
\begin{equation}
w_1(\rho,\omega)^2 = \lim_{a\to\infty}\frac{\left\|\rho_{\Lambda_a} - \omega_{\Lambda_a}\right\|_{W_1}^2}{\left|\Lambda_a\right|^2} \le \lim_{a\to\infty}\frac{S(\rho_{\Lambda_a}\|\omega_{\Lambda_a})}{2\left|\Lambda_a\right|} = \frac{s(\rho\|\omega)}{2}\,.
\end{equation}

Let $\Phi\in\mathcal{B}^r_{\mathbb{Z}^d}$ contain only terms acting on single sites.
We have for any $\Lambda\in\mathcal{F}_{\mathbb{Z}^d}$
\begin{equation}
\omega^\Phi_{\Lambda} = \bigotimes_{x\in\Lambda}\frac{e^{-\Phi(x)}}{\mathrm{Tr}_xe^{-\Phi(x)}} = \bigotimes_{x\in\Lambda}\omega^\Phi_x\,,
\end{equation}
therefore there exists a product state $\omega\in\mathcal{S}^I_{\mathbb{Z}^d}$ such that $\omega_\Lambda = \omega^\Phi_\Lambda$ for any $\Lambda\in\mathcal{F}_{\mathbb{Z}^d}$.
We have $w_1(\rho,\omega) = w_1(\rho,\Phi)$ and $s(\rho\|\omega) = s(\rho\|\Phi)$, therefore \eqref{eq:Marton} implies \eqref{eq:TCI} with $c=1$.
The claim follows.
\end{proof}

\subsection{Local commuting interactions at high temperature}\label{sec:HT}
A more general setting where \eqref{eq:TCI} can be proved is the case of geometrically local commuting interactions, where each spin interacts with a finite number of spins and all the terms of the interaction mutually commute.

Ref. \cite{de2022quantum} proved the following TCI for Gibbs states of local commuting interactions on a finite lattice employing a quantum generalization of Ollivier's coarse Ricci curvature \cite{ollivier2009ricci,gao2021ricci}:

\begin{thm}[High-temperature TCI for local commuting interactions {\cite[Theorem 4 and Proposition 9]{de2022quantum}}]\label{thm:TCI}
Let $\Phi\in\mathcal{B}^r_{\mathbb{Z}^d}$ be geometrically local and commuting, \emph{i.e.}, each spin interacts with at most $N$ spins where
\begin{equation}
N = \left|\bigcup_{0\in X\in\mathcal{F}_{\mathbb{Z}^d}\,:\,\Phi(X)\neq0} X\right| < \infty\,,
\end{equation}
and $\left[\Phi(X),\,\Phi(Y)\right]=0$ for any $X,\,Y\in\mathcal{F}_{\mathbb{Z}^d}$.
Let
\begin{equation}\label{eq:inft}
M = \inf_{t\ge0}\left(\left(e^{\left\|\Phi\right\|_r}+1\right)\sqrt{1+t^2}\left\|\Phi\right\|_r q^{\frac{3+\sqrt{1+t^2}}{2}} e^{\left\|\Phi\right\|_r\left(2+\frac{\sqrt{1+t^2}}{2}\right)} + 2\left\|\Phi\right\|_r e^{2\left\|\Phi\right\|_r} + 4\,e^{-\pi t}\right)\,,
\end{equation}
and let us assume that
\begin{equation}
\kappa = 1 - \left(2N-1\right)\left(N-1\right)M > 0\,.
\end{equation}
Then, for any $\Lambda\in\mathcal{F}_{\mathbb{Z}^d}$ and any $\rho\in\mathcal{S}_{\Lambda}$ we have
\begin{equation}
\left\|\rho - \omega^\Phi_{\Lambda}\right\|_{W_1}^2 \le \frac{2\,N^2\left|\Lambda\right|}{\left(1-e^{-\kappa}\right)^2} \, S\left(\rho\left\|\omega^\Phi_{\Lambda}\right.\right)\,.
\end{equation}
In particular, $\Phi$ satisfies \eqref{eq:TCI} with
\begin{equation}
c = \frac{4\,N^2}{\left(1-e^{-\kappa}\right)^2}\,.
\end{equation}
\end{thm}

\begin{rem}
Choosing in \eqref{eq:inft}
\begin{equation}
t = \frac{\ln\frac{1}{\left\|\Phi\right\|_r}}{\pi + \frac{\ln q}{2}}\,,
\end{equation}
we get $M \le O\left(\left\|\Phi\right\|_r\right)$ for $\left\|\Phi\right\|_r\to0$.
\end{rem}

Another strategy to prove TCIs for quantum spin systems on a finite lattice is to prove that suitable local quantum Markov semigroups that have the Gibbs state as unique fixed point satisfy a modified logarithmic Sobolev inequality \cite{bardet2021entropy,bardet2021rapid,capel2020modified,de2022quantum}, which states that the semigroup contracts exponentially the relative entropy with respect to the Gibbs state.
Ref. \cite{capel2020modified} proved that above a critical temperature, the Gibbs states of commuting nearest-neighbor interactions satisfy a modified logarithmic Sobolev inequality.
Exploiting this result, Ref. \cite{de2022quantum} proved the following TCI for such Gibbs states:

\begin{thm}[High-temperature TCI for nearest-neighbor interactions {\cite[Theorem 5]{de2022quantum}}]\label{thm:TCI2}
Let $\Phi\in\mathcal{B}^r_{\mathbb{Z}^d}$ be a nearest-neighbor interaction, \emph{i.e.}, $\Phi(\Lambda)=0$ for any $\Lambda\in\mathcal{F}_{\mathbb{Z}^d}$ that contains at least two sites that are not neighboring.
Then, there exists a critical inverse temperature $\beta_c>0$ such that for any $0\le\beta<\beta_c$ there exists $c_\beta>0$ such that for any $a\in\mathbb{N}_+^d$ we have
\begin{equation}
\left\|\rho_{\Lambda_a} - \omega^{\beta\Phi}_{\Lambda_a}\right\|_{W_1}^2 \le \frac{c_\beta\left|\Lambda_a\right|}{2} \, S\left(\rho_{\Lambda_a}\left\|\omega^{\beta\Phi}_{\Lambda_a}\right.\right)\,.
\end{equation}
In particular, $\beta\,\Phi$ satisfies \eqref{eq:TCI} with $c=c_\beta$.
\end{thm}

\begin{cor}
All the interactions satisfying the hypotheses either of \autoref{thm:TCI} or of \autoref{thm:TCI2} have a unique equilibrium state.
\end{cor}

\begin{rem}
The uniqueness of the equilibrium states for all the interactions $\Phi\in\mathcal{B}^r_{\mathbb{Z}^d}$ such that $r>\log q$ and $\left\|\Phi\right\|_r < \frac{1}{2q}$ has been proved in \cite{frolich2015some}.
\end{rem}

\section{Perspectives}\label{sec:persp}

In this paper we have proposed a specific Wasserstein distance of order $1$ for quantum spin systems on infinite lattices.
We expect the proposed distance to be a powerful tool in the study of the statistical mechanics of quantum spin systems, quantum dynamical systems, and tomography of quantum states:

\begin{enumerate}
\item The specific quantum $W_1$ distance can be employed to study the diameter of the set of the equilibrium states of an interaction close to a thermal phase transition.
Above the critical temperature the equilibrium state is unique and the diameter is zero, while below the critical temperature the diameter is strictly positive.
The limit of the diameter of the set of the equilibrium states as the temperature tends from below to the critical value, and in particular whether such limit is zero or strictly positive, can be employed to characterize the phase transition.

\item In \cite{ornstein1973application}, Ornstein proposed the $\bar{d}$-distance as a natural metric for the classification of stochastic processes and singled out a large class (the so-called $B$-processes), containing e.g.\ all mixing Markov processes, such that a fundamental isomorphism theorem holds: two processes are isomorphic if and only if their entropies coincide. The isomorphism here is in the sense of dynamical systems, i.e., a measurable and invertible transformation mapping one probability measure to the other. The specific quantum $W_1$ distance could provide a useful analytical tool towards establishing analogous results in the setting of quantum dynamical systems \cite{alicki2001quantum}.

\item The statistical problem of estimating a stationary ergodic process, in a given family, from the observation of a single sample path of length $n$ can be quantitatively addressed using Ornstein's $\bar{d}$-distance. In \cite{ornstein1990sampling}, it is proved that an empirical block scheme, i.e., the product probability naturally obtained from the observed frequencies on sliding window of length $k$, converges in the $\bar{d}$-distance, as $n$ grows, towards the target process, provided that it is a $B$-process and $k$ grows at least logarithmically with respect to $n$. Exploring the quantum analogue of this and related results, e.g.\ for discrimination between two sampled processes \cite{ornstein1994d},  may extend the scope of the recent works \cite{rouze2021learning,maciejewski2021exploring} on tomography of a quantum state and stimulate novel approaches, particularly when the number of accessible independent copies is extremely constrained.
\end{enumerate}

\section*{Acknowledgements}
We thank Emily Beatty for useful suggestions to improve the presentation of the proof of \autoref{thm:main}.
GDP has been supported by the HPC Italian National Centre for HPC, Big Data and Quantum Computing – Proposal code CN00000013, CUP J33C22001170001, funded within PNRR - Mission 4 - Component 2 Investment 1.4.
GDP is a member of the ``Gruppo Nazionale per la Fisica Matematica (GNFM)'' of the ``Istituto Nazionale di Alta Matematica ``Francesco Severi'' (INdAM)''. DT is a member of the INdAM group ``Gruppo Nazionale per l'Analisi Matematica, la Probabilità e le loro Applicazioni (GNAMPA)'' and  was partially supported by the INdAM-GNAMPA project 2022 ``Temi di Analisi Armonica Subellittica''.

\appendix

\section{Properties of the quantum \texorpdfstring{$W_1$}{W\_1} distance}\label{app:W1}

\begin{prop}[{\cite[Proposition 2]{de2021quantum}}]\label{prop:W1T}
For any finite set $\Lambda$ and any $\Delta\in\mathcal{O}_\Lambda^T$ we have
\begin{equation}
\frac{1}{2}\left\|\Delta\right\|_1 \le \left\|\Delta\right\|_{W_1} \le \frac{\left|\Lambda\right|}{2}\left\|\Delta\right\|_1\,.
\end{equation}
\end{prop}

\begin{prop}[{\cite[Proposition 5]{de2021quantum}}]\label{prop:local}
Let $\Lambda'\subseteq\Lambda$ be finite sets.
Then, for any $\Delta\in\mathcal{O}_\Lambda^T$ such that $\mathrm{Tr}_{\Lambda'}\Delta=0$ we have
\begin{equation}
\left\|\Delta\right\|_{W_1} \le \frac{q^2-1}{q^2}\left|\Lambda'\right|\left\|\Delta\right\|_1\,.
\end{equation}
\end{prop}

\begin{prop}[Superadditivity {\cite[Proposition 4]{de2021quantum}}]\label{prop:W1SA}
The quantum $W_1$ distance is superadditive in general and additive for product states, \emph{i.e.}, for any two disjoint finite sets $\Lambda,\,\Lambda'$ and any $\rho,\,\sigma\in\mathcal{S}_{\Lambda\Lambda'}$ we have
\begin{equation}
\left\|\rho - \sigma\right\|_{W_1} \ge \left\|\rho_\Lambda - \sigma_\Lambda\right\|_{W_1} + \left\|\rho_{\Lambda'} - \sigma_{\Lambda'}\right\|_{W_1}\,,
\end{equation}
and for any $\rho_\Lambda,\,\sigma_\Lambda\in\mathcal{S}_\Lambda$ and any $\rho_{\Lambda'},\,\sigma_{\Lambda'}\in\mathcal{S}_{\Lambda'}$ we have
\begin{equation}
\left\|\rho_\Lambda\otimes\rho_{\Lambda'} - \sigma_\Lambda\otimes\sigma_{\Lambda'}\right\|_{W_1} = \left\|\rho_\Lambda - \sigma_\Lambda\right\|_{W_1} + \left\|\rho_{\Lambda'} - \sigma_{\Lambda'}\right\|_{W_1}\,.
\end{equation}
\end{prop}

\section{Auxiliary proofs}\label{app:auxproofs}

\subsection{Proof of \autoref{prop:TI}}\label{sec:propTI}
\begin{prop*}[\ref{prop:TI}]
The trace distance on $\mathcal{S}_{\mathbb{Z}^d}$ is the supremum of the trace distances between the marginal states: For any $\rho,\,\sigma\in\mathcal{S}_{\mathbb{Z}^d}$,
\begin{equation}
T(\rho,\sigma) = \frac{1}{2}\sup_{\Lambda\in\mathcal{F}_{\mathbb{Z}^d}}\left\|\rho_\Lambda - \sigma_\Lambda\right\|_1\,,
\end{equation}
where $\|\cdot\|_1$ denotes the trace norm on $\mathfrak{U}_\Lambda$ given by
\begin{equation}
\left\|A\right\|_1 = \mathrm{Tr}_\Lambda\sqrt{A^\dag A}\,,\qquad A\in\mathfrak{U}_\Lambda\,.
\end{equation}
\end{prop*}

\begin{proof}
Since $\mathfrak{U}_{\mathbb{Z}^d}^{loc}$ is dense in $\mathfrak{U}_{\mathbb{Z}^d}$, we have
\begin{align}
2\,T(\rho,\sigma) &= \sup_{A\in\mathfrak{U}_{\mathbb{Z}^d}^{loc}:\|A\|_\infty\le1}\left|\rho(A) - \sigma(A)\right| = \sup_{\Lambda\in\mathcal{F}_{\mathbb{Z}^d}}\sup_{A\in\mathfrak{U}_\Lambda:\|A\|_\infty\le1}\left|\rho(A) - \sigma(A)\right|\nonumber\\
&= \sup_{\Lambda\in\mathcal{F}_{\mathbb{Z}^d}}\sup_{A\in\mathfrak{U}_\Lambda:\|A\|_\infty\le1}\left|\mathrm{Tr}_\Lambda\left[\left(\rho_\Lambda - \sigma_\Lambda\right)A\right]\right| = \sup_{\Lambda\in\mathcal{F}_{\mathbb{Z}^d}}\left\|\rho_\Lambda - \sigma_\Lambda\right\|_1\,.
\end{align}
The claim follows.
\end{proof}

\subsection{Proof of \autoref{prop:equivpartial}}\label{sec:equivpartial}
\begin{prop*}[\ref{prop:equivpartial}]
For any $\Lambda\in\mathcal{F}_{\mathbb{Z}^d}$, any $H\in\mathcal{O}_\Lambda$ and any $x\in\Lambda$, \eqref{eq:partialxH} and \eqref{eq:partialxHI} are equivalent.
\end{prop*}

\begin{proof}
Let
\begin{equation}
\partial_x H = 2\min_{A\in \mathcal{O}_{\Lambda\setminus x}}\left\|H - A\right\|_\infty\,,\qquad \tilde{\partial}_x H = 2\inf_{A\in \mathcal{O}_{\mathbb{Z}^d\setminus x}}\left\|H - A\right\|_\infty\,.
\end{equation}
We clearly have $\tilde{\partial}_x H \le \partial_x H$.
Let $\omega_{\mathbb{Z}^d\setminus\Lambda}\in \mathcal{S}_{\mathbb{Z}^d\setminus\Lambda}$ be the uniform distribution on $\mathbb{Z}^d\setminus\Lambda$, and let $\Psi_\Lambda:\mathfrak{U}_{\mathbb{Z}^d}\to\mathfrak{U}_\Lambda$ be the completely positive unital linear map such that for any $A\in\mathfrak{U}_{\mathbb{Z}^d}$ and any $\rho_\Lambda\in\mathcal{S}_\Lambda$
\begin{equation}
\mathrm{Tr}_\Lambda\left[\rho_\Lambda\,\Psi_\Lambda(A)\right] = (\omega_{\mathbb{Z}^d\setminus\Lambda}\otimes\rho_\Lambda)(A)\,.
\end{equation}
Let $A\in\mathcal{O}_{\mathbb{Z}^d\setminus x}$.
We have for any $\rho_\Lambda\in\mathcal{S}_\Lambda$ and any unitary operator $U_x\in\mathfrak{U}_x$
\begin{align}
\mathrm{Tr}_\Lambda\left[\rho_\Lambda\,U_x^\dag\,\Psi_\Lambda(A)\,U_x\right] &= \mathrm{Tr}_\Lambda\left[U_x\,\rho_\Lambda\,U_x^\dag\,\Psi_\Lambda(A)\right] = \left(\omega_{\mathbb{Z}^d\setminus\Lambda}\otimes U_x\,\rho_\Lambda\,U_x^\dag\right)(A)\nonumber\\
&= (\omega_{\mathbb{Z}^d\setminus\Lambda}\otimes\rho_\Lambda)\left(U_x^\dag\,A\,U_x\right)= (\omega_{\mathbb{Z}^d\setminus\Lambda}\otimes\rho_\Lambda)(A) = \mathrm{Tr}_\Lambda\left[\rho_\Lambda\,\Psi_\Lambda(A)\right]\,,
\end{align}
therefore $U_x^\dag\,\Psi_\Lambda(A)\,U_x = \Psi_\Lambda(A)$, hence $\Psi_\Lambda(A)\in\mathcal{O}_{\Lambda\setminus x}$.
We then have
\begin{equation}
\partial_x H \le 2\left\|H - \Psi_\Lambda(A)\right\|_\infty = 2\left\|\Psi_\Lambda(H-A)\right\|_\infty \le 2\left\|H-A\right\|_\infty\,,
\end{equation}
where the last inequality follows since $\Psi_\Lambda$ is completely positive and unital.
We then have $\partial_x H \le \tilde{\partial}_x H$.
The claim follows.
\end{proof}

\section{Auxiliary lemmas}\label{app:aux}

\begin{lem}[Multidimensional Fekete's lemma \cite{capobianco2008multidimensional}]\label{lem:Fekete}
Let $f:\mathbb{N}_+^d\to\mathbb{R}$ be superadditive with respect to each variable, \emph{i.e.},
\begin{equation}
f(x_1,\,\ldots,\,x_i+t,\,\ldots,\,x_d) \ge f(x_1,\,\ldots,\,x_i,\,\ldots,\,x_d) + f(x_1,\,\ldots,\,t,\,\ldots,\,x_d)
\end{equation}
for any $x_1,\,\ldots,\,x_d,\,t\in\mathbb{N}$ and any $i=1,\,\ldots,\,d$.
Then,
\begin{equation}
\lim_{x\to\infty}\frac{f(x)}{x_1\ldots x_d} = \sup_{x\in\mathbb{N}_+^d}\frac{f(x)}{x_1\ldots x_d}\,.
\end{equation}
\end{lem}

\begin{lem}\label{lem:PSD}
Let $H\in\mathcal{O}_\Lambda$ be positive semi-definite.
Then, for any $x\in\Lambda$,
\begin{equation}
\partial_x H \le \left\|H\right\|_\infty\,.
\end{equation}
\end{lem}

\begin{proof}
We have
\begin{equation}
-\frac{\left\|H\right\|_\infty}{2}\,\mathbb{I} \le H - \frac{\left\|H\right\|_\infty}{2}\,\mathbb{I} \le \frac{\left\|H\right\|_\infty}{2}\,\mathbb{I}\,,
\end{equation}
therefore
\begin{equation}
\partial_x H \le 2\left\|H - \frac{\left\|H\right\|_\infty}{2}\,\mathbb{I}\right\|_\infty \le \left\|H\right\|_\infty\,.
\end{equation}
The claim follows.
\end{proof}

\begin{prop}\label{prop:Marton}
Let $\Lambda_1,\,\ldots,\,\Lambda_k$ be $k$ copies of the finite set $\Lambda$.
Then, for any $\rho\in\mathcal{S}_{\Lambda_1\ldots\Lambda_k}$ and any $\sigma\in\mathcal{S}_\Lambda$ we have
\begin{equation}
\left\|\rho - \sigma^{\otimes k}\right\|_{W_1}^2 \le 2k\left|\Lambda\right|^2\,S\left(\rho\left\|\sigma^{\otimes k}\right.\right)\,.
\end{equation}
\end{prop}
\begin{proof}
The proof follows the same lines as the proof of \cite[Theorem 2]{de2021quantum}.
We have
\begin{align}
\left\|\rho - \sigma^{\otimes k}\right\|_{W_1} &\le \sum_{i=1}^k\left\|\sigma^{\otimes\left(i-1\right)}\otimes\rho_{\Lambda_i\ldots\Lambda_k} - \sigma^{\otimes i}\otimes\rho_{\Lambda_{i+1}\ldots\Lambda_k}\right\|_{W_1} \nonumber\\
&\overset{\mathrm{(a)}}{\le} \left|\Lambda\right|\sum_{i=1}^k\left\|\rho_{\Lambda_i\ldots\Lambda_k} - \sigma\otimes\rho_{\Lambda_{i+1}\ldots\Lambda_k}\right\|_1 \overset{\mathrm{(b)}}{\le} \left|\Lambda\right|\sum_{i=1}^k\sqrt{2\,S\left(\rho_{\Lambda_i\ldots\Lambda_k}\left\|\sigma\otimes\rho_{\Lambda_{i+1}\ldots\Lambda_k}\right.\right)} \nonumber\\
&= \left|\Lambda\right|\sum_{i=1}^k\sqrt{2\left(S(\rho_{\Lambda_i}) + S(\rho_{\Lambda_{i+1}\ldots\Lambda_k}) - S(\rho_{\Lambda_i\ldots\Lambda_k}) + S(\rho_{\Lambda_i}\|\sigma)\right)} \nonumber\\
&\overset{\mathrm{(c)}}{\le} \left|\Lambda\right|\sqrt{2k\sum_{i=1}^k\left(S(\rho_{\Lambda_i}) + S(\rho_{\Lambda_{i+1}\ldots\Lambda_k}) - S(\rho_{\Lambda_i\ldots\Lambda_k}) + S(\rho_{\Lambda_i}\|\sigma)\right)} \nonumber\\
&=\left|\Lambda\right|\sqrt{2k}\sqrt{\sum_{i=1}^k\left(S(\rho_{\Lambda_i}) + S(\rho_{\Lambda_i}\|\sigma)\right)- S(\rho)} =\left|\Lambda\right|\sqrt{2k\,S\left(\rho\left\|\sigma^{\otimes k}\right.\right)}\,.
\end{align}
(a) follows from \autoref{prop:local} observing that
\begin{equation}
\mathrm{Tr}_{\Lambda_i}\left[\sigma^{\otimes\left(i-1\right)}\otimes\rho_{\Lambda_i\ldots\Lambda_k} - \sigma^{\otimes i}\otimes\rho_{\Lambda_{i+1}\ldots\Lambda_k}\right] = 0\,;
\end{equation}
(b) follows from Pinsker's inequality; (c) follows from the concavity of the square root.
The claim follows.
\end{proof}

\begin{lem}\label{lem:Wk}
We have
\begin{equation}
\dim\mathcal{W}_k \le D_k\dim\mathcal{V}\,.
\end{equation}
\end{lem}

\begin{proof}
Let $A_0,\,\ldots,\,A_{q^2-1}$ be a basis of $\mathbb{C}^{q\times q}$ with $A_0 = \mathbb{I}$.
For any $x\in\left\{0,\,\ldots,\,q^2-1\right\}^{\Lambda}$, let
\begin{equation}
A_x = \bigotimes_{i\in\Lambda}A_{x_i}\,,
\end{equation}
where each $A_{x_i}$ acts on the site $i$.
We have
\begin{equation}
\mathcal{W}_k \subseteq \mathrm{span}\left\{A_x|\psi\rangle:|\psi\rangle\in\mathcal{V},\,H(x)\le k\right\}\,.
\end{equation}
We also have
\begin{equation}
\mathcal{W}_{k-1} \subseteq \mathrm{span}\left\{A_x|\psi\rangle:|\psi\rangle\in\mathcal{V},\,H(x)\le k-1\right\}\,,
\end{equation}
and since $\mathcal{W}_k\perp\mathcal{W}_{k-1}$, we have
\begin{equation}
\mathcal{W}_k \subseteq \mathrm{span}\left\{A_x|\psi\rangle:|\psi\rangle\in\mathcal{V},\,H(x) = k\right\}\,.
\end{equation}
Therefore,
\begin{equation}
\dim\mathcal{W}_k \le \left|H^{-1}(k)\right|\dim\mathcal{V}\,.
\end{equation}
The claim follows.
\end{proof}

\begin{lem}\label{lem:partialxH}
Let $\Phi\in\mathcal{B}_{\mathbb{Z}^d}^r$.
Then, for any $\Lambda\in\mathcal{F}_{\mathbb{Z}^d}$ and any $x\in\Lambda$ we have
\begin{equation}
\partial_x H^\Phi_\Lambda \le 2\left\|\Phi\right\|_r\,,
\end{equation}
and
\begin{equation}\label{eq:LL}
\left\|H^\Phi_\Lambda\right\|_L \le 2\left\|\Phi\right\|_r\,.
\end{equation}
\end{lem}

\begin{proof}
We have
\begin{equation}
\partial_x H^\Phi_\Lambda \le 2\sum_{x\in X \subseteq\Lambda}\left\|\Phi(X)\right\|_\infty \overset{\mathrm{(a)}}{=} 2\sum_{0\in X \subseteq\Lambda-x}\left\|\Phi(X)\right\|_\infty \le 2\sum_{0\in X \in \mathcal{F}_{\mathbb{Z}^d}}\left\|\Phi(X)\right\|_\infty \le 2\left\|\Phi\right\|_r\,,
\end{equation}
where (a) follows from the translation invariance of $\Phi$.
The claim follows.
\end{proof}

\bibliography{lattices}
\bibliographystyle{unsrt}
\end{document}